\documentclass[12pt]{article}
\usepackage{mynote}
\newcommand{\blind}{0}

\begin{document}

\def\spacingset#1{\renewcommand{\baselinestretch}%
{#1}\small\normalsize} \spacingset{1}


\if0\blind
{
  \title{\bf Inference on the Significance of Modalities in Multimodal Generalized Linear Models}
  \author{Wanting
    Jin$^1$, Guorong Wu$^2$, and Quefeng Li$^3$\\
    $^{1,3}$ Department of Biostatistics, Gillings School of Global Public Health, \\
    University of North Carolina at Chapel Hill\\
  $^2$ Department of Psychiatry, \\ University of North Carolina at Chapel Hill} 
  \date{}
  \maketitle
} \fi

\if1\blind
{
  \bigskip
  \bigskip
  \bigskip
  \begin{center}
    {\LARGE\bf Inference on the Significance of Modalities in Multimodal Generalized Linear Models}
\end{center}
  \medskip
} \fi

\bigskip
\begin{abstract}
 Despite the popular of multimodal statistical models, there lacks rigorous statistical inference
  tools for inferring the significance of a single modality within a multimodal model, especially in high-dimensional models. For high-dimensional multimodal generalized
  linear models, we propose a novel entropy-based metric, called the expected relative entropy, to quantify the information
  gain of one modality in addition to all other modalities in the model. We propose a deviance-based statistic to
  estimate the expected relative entropy, prove that it is consistent and its asymptotic distribution can be approximated by a non-central chi-squared
  distribution. That enables the calculation of confidence intervals and p-values to assess the significance of the expected
  relative entropy for a given modality. We numerically evaluate the empirical performance of our proposed inference tool by simulations and apply it to a multimodal neuroimaging dataset to demonstrate its good performance on various high-dimensional multimodal generalized linear models.  
\end{abstract}

\noindent%
\vfill

\newpage
\spacingset{1.9} 
\section{Introduction}
\label{sec:intro}
Multimodal data, which collect information from multiple sources or modalities of the same subjects, are fast emerging
in many areas, including artificial intelligence, biomedical research, finance, social science, etc. For example,
multimodal neuroimaging employs diverse imaging technologies such as Magnetic Resonance Imaging (MRI), Positron Emission
Tomography (PET) and Computed Tomography (CT) scan to study structural and functional changes in the brain. As different
imaging modalities characterizes different aspects of the brain, recent neuroimaging studies \citep{Uludag2014} have
demonstrated that fused multimodal models often outperform single-modal models in predicting disease phenotypes and
understanding disease progression mechanisms. Multi-omics data are another notable example of multimodal data in
biomedical research. \cite{Richardson2016} illustrated that the integrative analysis of multi-omics data can enhance our
understanding of the genomic mechanisms underlying various diseases. In recent statistics literature, many new
integrative models have been proposed for multimodal data analysis, including unsupervised
\citep{gaynanova2019structural,li2017incorporating,lock2013joint} and supervised learning methods
\citep{Quefeng2022,xue2021integrating}. There are also theoretical studies proving that, under certain conditions,
multimodal models are guaranteed to have better statistical performance than unimodal models \citep{li2018integrative}.

In multimodal data analysis, a critical yet challenging problem is on how to quantify the information contributed by
individual modalities within a multimodal model. This problem is particularly relevant in various scientific domains. For
instance, in genetics, heritability is a crucial concept that aims to determine the proportion of phenotypic variability
attributable to genetic variation, as opposed to environmental, behavioral, or other factors. In neuroscience, when treating
patients with dementia, physicians often employ multiple imaging modalities for diagnosis based on the patient's specific
condition. Therefore, assessing the informational value of various imaging modalities in the disease phenotype is crucial
for devising accurate treatment plans and preventing unnecessary diagnostic procedures. These practical problems raise the
great need for rigorous statistical methods to quantify the information contributed by individual modalities and assess its
statistical significance in multimodal statistical models.

In classical statistics literature, there were many metrics quantifying the goodness-of-fit of a statistical model. These
metrics can be regarded as the quantification of information contributed by all variables in a unimodal statistical
model. For instance, $R^2$ and adjusted-$R^2$ were used to quantify the proportion of variation explained by all variables
in a linear regression. These metrics were further extended to pseudo-$R^2$ metrics in the context of
generalized linear model \citep{Zheng2000}. \cite{Hu2006} studied the property of the
pseudo-$R^2$ metric under multinomial regression and proved that its limit can be interpreted as the relative entropy
between the full and the null models.
Unfortunately, these classic metrics are inadequate to measure and infer the contribution of individual modalities in
multimodal models. The reason is that multimodal models are often high-dimensional, as each modality can contain a large
number of variables. Consequently, inferring the significance of modalities constitutes a high-dimensional inference
problem. In recent years, many high-dimensional inference methods have been developed. For example,
\cite{javanmard2014confidence,Ning2017,VarderGeer2014,zhang2014confidence} have developed de-biased or de-sparsifying least
absolute shrinkage and selection operator (LASSO) methods and de-correlated score tests for testing the significance of a
subvector of the covariate coefficients in unimodal regression models. \cite{zhang2017simultaneous} further developed a test
for testing the whole covariate coefficients using the max norm of the desparsifying estimators.  \cite{shi2019linear} and
\cite{ZhuBradic2018} developed partially penalized methods for testing an arbitrary linear combination of covariate
coefficients in a sparse generalized linear model and a dense high-dimensional linear regression model.
Despite these advancements, existing inference methods are limited to the variable level and cannot be
directly applied to our specific problem of evaluating and comparing high-dimensional modalities.  This is because the
contribution of individual modalities depends not only on the magnitude of covariate coefficients, but also other population
parameters such as the correlation between different modalities. Therefore, inferring such a quantity needs to infer the
significance of a complicated function of covariate coefficients and some other population parameters, making the task more
challenging than inferring covariate coefficients alone. 

Recent years have also witnessed advancements in the inference of
covariate coefficient functions. For example, 
\cite{tony2020semisupervised} developed methods to infer the explained variance in a high-dimensional linear model, which is
a quadratic function of the covariate coefficients.  As to be discussed in Section \ref{sec:model}, our proposed metric is
more complicated than a quadratic function, which poses a more challenging high-dimensional inference problem. Besides, we
consider the inference problem in the context of multimodal models, instead of unimodal models.

In this article, we propose an entroy-based metric to quantify the contribution of individual modalities in multimodal
generalized linear models (MGLM). We justify such a metric by studying its mathematical properties and establishing connections with
conventional goodness-of-fit metrics for specific GLMs. We show that those conventional metrics can be
treated as special cases of our proposed metric. Then, we develop a deviance-based estimator for this metric and study its
statistical properties. We show that the estimator is consistent and provide an error bound. Moreover, we derive the
limiting distribution of the estimator, enabling us to construct confidence intervals and calculate p-values to assess the
significance of individual modalities. We show that such an asymptotic distributional result does not rely on variable
selection consistency, which makes it more flexible for high-dimensional models. To the best of our knowledge,
these inference tools are among the first to perform rigorous statistical inference at the modality level in the context of
MGLM. They address the critical need of inferring the significance of modalities in
multimodal models and can be applied to a wide range of multimodal data analysis problems.

The rest of the article is organized as follows. In Section \ref{sec:model}, we propose to use the Expected Relative Entropy
(ERE) as a metric for quantifying information contributed by individual modalities. We study its mathematical properties and
relate it to conventional goodness-of-fit metrics for unimodal regression models. We also propose a two-step, deviance-based
estimator for this metric. In Section \ref{sec:statistical-results}, we give the error bound of the estimator and derive its
asymptotic distribution, upon which the confidence interval and the p-value for the significance of a modality can be
calculated. We use simulation studies to assess the numerical performance of the proposed metric and the corresponding
estimator in Section \ref{sec:Num}. In Section \ref{sec:Num-neuro}, we apply our proposed method to analyze a multi-modal neuroimaging dataset to demonstrate
how it can applied to a real-world problem. We conclude our article with discussions on how our methods can be extended to
other multimodal regression models in Section \ref{sec:discuss}.

\section{Expected Relative Entropy}
\label{sec:model}
\subsection{Model Setup}
\label{sec:model-setup} 
Suppose there are $M$ modalities of variables. Let $\x_m \in \Rcal^{p_m}$ denote the vector of $p_m$ variables from the
$m$-th modality, $\x = (\x_{1}^T,\dots,\x_{M}^T)^T \in \Rcal^{p}$ and $p = \sum^{M}_{m=1}p_m$. Let $y$ denote the response
variable.  We assume that $y$ relates to $\x$ through a Multimodal Generalized Linear Model that the conditional
distribution of $y$ given $\x$ belongs to the canonical exponential family, which has the following density function
\begin{equation}
  \label{eq:1}
  p(y|\x) = \exp\{\phi^{-1} [ y\x^T\betabs - b(\x^T\betabs) ] + c(y)\},
\end{equation}
where $\betabs = (\betab_1^{*^T}, \dots \betab_M^{*^T})^T \in \Rcal ^p$ is the vector of true regression coefficients,
$\betab_m^* \in \Rcal^{p_m}$ is the vector of true regression coefficients in the $m$-th modality, $b(\cdot)$ and
$c(\cdot)$ are known functions, and $b(\cdot)$ is a twice continuously differentiable function with a positive second
derivative.  We assume the dispersion parameter $\phi=1 $ for simplicity. Such an assumption was also used other
high-dimensional GLM inference works \citep{Ning2017,ouyang2023high,Li02042024}. When $\phi$ is unknown in practice, we
can replace it with some consistent estimator. As proved in some previous works \citep{Quefeng2022}, the asymptotic
distribution of the resulting test statistic remains the same once such a consistent estimator is used.

Next, we propose a metric to quantify the information gain of the $m$-th modality in the MGLM. Let
$\calM=\{j:\beta^{*}_j\neq 0 \}$ be the set of indices for non-zero elements of $\betabs$ and $\calM_{-m}$ be the subset of
$\calM$ with indices outside the $m$-th modality. Excluding the $m$-th modality, we define an oracle working model that
regress $y$ on $\x_{\calM_{-m}}$, the subvector of $\x$ outside the $m$-th modality with nonzero coefficients, and let
\begin{equation}
  \label{eq:19}
\betab_{0_{\calM_{-m}}}=\argmax_{\betab_{\calM_{-m}}}~\E[\exp\{\phi^{-1} [ y\x_{\calM_{-m}}^T\betab_{\calM_{-m}} - b(\x_{\calM_{-m}}^T\betab_{\calM_{-m}}) ] + c(y)\}].
\end{equation}
In other words, $\betab_{0_{\calM_{-m}}}$ can be treated as the pseudo true coefficients for the oracle working model that
utilizes the true predictors outside the $m$-th modality. Finally, we expand $\betab_{0_{\calM_{-m}}}$ to be
$\betab_0=(\betab_{0_{\calM_{-m}}}^T,\zero^T)^T\in \mathcal{R}^p$ and define
\begin{equation*}
  p(y|\x_{-m})=\exp\{\phi^{-1} [ y\x^T\betab_0 - b(\x^T\betab_0) ] + c(y)\},
\end{equation*}
where $\x_{-m}$ is the $(p-p_m)$-dim subvector of $\x$ without variables in the $m$-th modality. 
The Kullback--Leibler (KL) divergence, also known as the relative entropy, measuring the disparity
between $p(y|\x)$ and $p(y|\x_{-m})$, is defined as
\begin{equation*}
D_{KL} ( p(y|\x)\|p(y|\x_{-m})) = \int_{y \in \mathcal{Y}}p(y|\x)\log \{{p(y|\x)}/{p(y|\x_{-m})} \}dy,
\end{equation*}
where $\mathcal{Y}$ is the sample space of $y$. The KL divergence has been widely used in statistics and information theory
\citep{Amari2016}. Our proposed KL divergence quantifies the entropy difference between the full
and reduced models, and is a random quantity depending on $\x$. To quantify the information gain of the
$m$-th modality, we propose to use
\begin{equation}
  \label{eq:3}
  H_m = 2\E_{\x}[D_{\text{KL}}(p(y|\x)\parallel p(y|\x_{-m}))],
\end{equation}
where the expectation is taken with respect to the distribution of $\x$. In other words, $H_m$ up to a factor of 2,
represents the Expected Relative Entropy (ERE) of models with and without the $m$-th modality. The factor of 2 is introduced
here to simplify the writing of the asymptotic results in Theorem \ref{thm:2}.
  Our paper's main goal is to develop valid tools to perform statistical inference on $H_m$.  As will be demonstrated
  in several concrete examples below, $H_m$ depends not only on the magnitude of coefficients, but also on other population
  parameters. Therefore, existing high-dimensional inference methods
  \citep{Ning2017,shi2019linear,VarderGeer2014,zhang2014confidence} cannot be directly applied to perform inference on
  $H_m$, which motivates us to develop new inference tools. In below, we will first justify why $H_m$ is an appropriate
metric for measuring the information gain of the $m$-th modality and connect it with existing goodness-of-fit measures for
conventional unimodal GLMs. In this way, our proposed metric can be viewed as a natural extension of
those existing measurements.  Next, we study some mathematical properties of $H_m$.

First, $H_m$ is non-negative, which is due to the non-negativity of the KL divergence. Second, Proposition \ref{pro1} shows that $H_m$ is monotone, which implies that the ERE contributed by
multiple modalities is always no less than that of one modality. The non-negativity and monotonicity properties demonstrate
that the ERE is an appropriate metric for quantifying the information gain of an individual modality in MGLM. Next, we calculate $H_m$ for some popular generalized linear models and show how it relates to the existing goodness-of-fit metrics.
\begin{pro}
  \label{pro1}
  If $\E|\log \{p(y|\x_{-m})/p(y|\x_{-(m,l)})\}| < \infty$, then it holds that $H_{(m,l)} \geq H_m$, where
  $H_{(m,l)}=2\E_{\x}[D_{\text{KL}}(p(y|\x)\parallel p(y|\x_{-(m,l)}))]$ and $\x_{-(m,l)}$ is the subvector for $\x$
  without variables in the $m$- and $l$-th modalities.
\end{pro}
\begin{ex}[Linear regression] \label{ex1}
  Suppose $ y = \sum_{m = 1}^{M} \x_m^T\betab_m^* + \epsilon$, where
  $\epsilon$ is the random error with $\E(\epsilon) = 0$ and $\mathrm{Var}(\epsilon) = \sigma^2_{\epsilon}$. If we further
  assume $\x \sim N_p(0, \Sigmab)$, then it follows that
  \begin{equation}
    \label{eq:example1}
    H_m = 
    \log\frac{\sigma_{m|-m}^2 + \sigma^2_{\epsilon} }{\sigma^2_{\epsilon}} + 
    \frac{\betab_m^{*T} \Sigmab_m \betab_m^* - \sigma_{m|-m}^2 }{\sigma_{m|-m}^2  + \sigma^2_{\epsilon}},
  \end{equation}
  where $\Sigmab_m = \mathrm{Cov}(\x_m)$ and $\sigma^2_{m|-m} = \mathrm{Var}(\x_m^T\betab_m^*|\x_{-m})$. The proof of
  (\ref{eq:example1}) can be found in the Appendix.
\end{ex}

  As shown in (\ref{eq:example1}), $H_m$ depends not only on the coefficient $\betabs_m$, but also on the covariance matrix
  $\Sigmab_m$, the conditional variance $\sigma^{2}_{m|-m}$, and $\sigma_{\epsilon}^2$. Thus, performing
  inference on $H_m$ presents greater challenges than inference on $\betabs_m$ alone. \cite{Quefeng2022} considered the
multimodal linear regression model and proposed to use the conditional variance $\sigma^2_{m|-m}$ to quantify the
contribution of the $m$-th modality in that model. They showed that $\sigma^2_{m|-m}$ can be interpreted as the improvement
of the prediction error or the increment of the adjusted $R^2$ that is attributed to the inclusion of the $m$-th modality in
addition to other modalities; see Section 6 of \cite{Quefeng2022}. Equation \eqref{eq:example1} gives a closed-form
expression describing the relationship between $\sigma^2_{m|-m}$ and our proposed metric $H_m$.

For a MGLM with the density function shown in (\ref{eq:1}), the KL divergence is
\begin{equation}
    \label{eq:gD_KL}
    D_{\text{KL}}(p(y|\x)\parallel p(y|\x_{-m})) = b'(\x^T\betabs)\x^T(\betabs - \betab_0) + b(\x^T\betab_0) -
    b(\x^T\betabs). 
\end{equation}
Using (\ref{eq:gD_KL}), we give more concrete forms of $H_m$ for logistic, exponential and Poisson
regressions. 

\begin{ex}[Logistic regression] \label{ex2}  
  For logistic regression, the conditional density of $y$ given $\x$ is $p(y|\x) = \exp[ y\x^T\betabs - \log \{1 +
  \exp(\x^T\betabs) \} ]$
  so that $b(\theta) = \log \{1 + \exp(\theta) \}$. It follows from \eqref{eq:gD_KL} that 
  \begin{align} \label{eq: Hm_log1}
    H_m &= 2 \E_{\x} \left\{\x^T(\betabs - \betab_0) \frac{\exp(\x^T\betabs)}{1 + \exp(\x^T\betabs)}  - \log \frac{1 +
          \exp(\x^T\betab_0)}{1 + \exp(\x^T\betabs)} \right\}. 
  \end{align}
  Next, we discuss how $H_m$ relates to some conventional goodness-of-fit metric of logistic regression. We observe that
  $p(y|\x)$ can be also expressed as
$
        p(y|\x) = \exp \{ I(y=1)\log P_1(\x) + I(y=0)\log P_0(\x)\},
$
where $P_1(\x) = P(y=1|\x) = \exp(\x^T\betabs)/ \{1 + \exp(\x^T\betabs) \}$,
$P_0(\x) = P(y=0|\x) = 1/\{1 + \exp(\x^T\betabs) \}$. 
If we view all covariates as one single modality, the reduced model without any covariate follows the Bernoulli distribution
with $P(y=1)=\E_{\x}\{P_1(\x)\}$, which is the marginal distribution of $y$ denoted by $p(y)$. Then, we have $\log p(y) = I(y=1) \log P(y=1) + I(y=0) \log P(y=0)$. Following our definition in
(\ref{eq:3}), the ERE of all variables is given by
\begin{equation}
  \label{eq: Hm_log2}
  \begin{split}
    H_{all} &= 2\E_{\x}[D_{\text{KL}}(p(y|\x)\parallel p(y)) ]
              = 2 [ \E_{\x}\{\log p(y|\x)\} - \E_{\x}\{\log p(y)\} ] \\
           &= 2 \left[\sum_{k=0}^1 \E_{\x}\{P_k(\x) \log P_k(\x)\} - \sum_{k=0}^1\E_{\x}\{P_k(\x)\}\log E_{\x}\{P_k(\x)\}
             \right].
  \end{split}
\end{equation}
In (\ref{eq: Hm_log2}), the term $\sum_{k=0}^1 \E_{\x}\{P_k(\x) \log P_k(\x)\}$ can be viewed as the
  entropy of the conditional distribution of $y$ given $\x$, while the term
  $\sum_{k=0}^1\E_{\x}\{P_k(\x)\}\log \E_{\x}\{P_k(\x)\}$ represents the entropy of the marginal distribution of $y$.  In
this case, we can verify that $H_{all}$ is the same as the entropy-based metric proposed by \cite{Hu2006} for quantifying
the goodness-of-fit of the conventional logistic regression model.  
Consequently, Hu's entropy-based metric can be
considered as a special case of our ERE metric in a classic unimodal setting that regards all covariates as a single
modality. 
\end{ex}

\begin{ex}[Exponential Regression] \label{ex3} For the exponential regression, the conditional density has the form of
  $ p(y|\x) = \exp \{ - y \x^T\betabs + \log (\x^T\betabs) \}$ so that $b(\theta) = - \log (\theta)$. Therefore, 
  \begin{equation*}
    H_m  = 2 \E_{\x} \left\{\log \frac{\x^T\betabs}{\x^T\betab_0} - \left(1 - \frac{\x^T\betab_{0}}{\x^T\betabs}\right)
    \right\}. 
    \end{equation*}
\end{ex}

\begin{ex}[Poisson Regression] \label{ex4} For the Poisson regression, the conditional density has the form of
  $p(y|\x) = \exp \{ y \x^T\betabs - \exp (\x^T\betabs) - \log (y!)\}$ so that $b(\theta) = \exp(\theta)$. Therefore, 
\begin{align*}
    \label{eq:D_KL_Poisson}
  H_m = 2 \E_{\x}\{\exp(\x^T\betabs)\x^T (\betabs-\betab_0)  + \exp(\x_{}^T\betab_0) - \exp(\x^T\betabs) \}. 
\end{align*}
\end{ex}

Unlike linear regression, the closed forms of $H_m$ for logistic, exponential and Poisson regressions are hard to be derived
due to their non-linear link functions, even when assuming $\x$ is multi-variate normal. However, in simulation studies,
when the distribution of $\x$ is known, we can use the Markov Chain Monte Carlo method to approximate $H_m$. More details of
such an approximation will be given in Section \ref{sec:Num}. 

\subsection{ Estimation of the Expected Relative Entropy}
\label{sec:estim-expect-relat}

In practice, suppose we have independent and identically distributed (i.i.d) samples $\{(Y_i, \X_i)\}_{i=1}^n$ from the
distribution of $(y,\x)$. We propose a deviance-based estimator of $H_m$, which is motivated by a seminal work
\citep{Hastie1987} that explored the relationship between the deviance and the KL divergence. The deviance of the full model
using all modalities is $D(\betab)=2\sum_{i=1}^n \{\log p(Y_i) - \log p(Y_i|\X_i;\betab)\}$, where $p(Y_i)$ denotes the
marginal density of $y$ and $p(Y_i|\X_i;\betab)$ denotes the conditional density of $y|\x$ with the parameter
$\betab\in \mathcal{R}^p$. The deviance of the reduced model using all but the $m$-th modality is
$D_{-m}(\betab_{0})=2\sum_{i=1}^n \{\log p(Y_i) - \log p(Y_i|\X_{i};\betab_0)\}$, where $p(Y_i|\X_i;\betab_0)$ denotes the
conditional density of $y|\x_{-m}$ with the parameter $\betab_{0_{-m}} \in \mathcal{R}^{p-p_m}$ in the reduced model and
$\betab_0=(\betab^T_{0_{-m}},\zero^T)^T\in \mathcal{R}^p$ being its expansion to a $p$-dim vector. Here, we slightly abuse
the notation $\betab_0$ by denoting it as an arbitrary parameter vector in the reduced model. Its true value is defined in
(\ref{eq:19}) using the same notation. \cite{Hastie1987} showed that the average deviance difference
$n^{-1}\{D(\betab)-D_{-m}(\betab_0) \}$ can be treated as a sample analog of the KL divergence
$D_{\text{KL}}(p(y|\x)\parallel p(y|\x_{-m}))$. This motivates us to estimate $H_m$ by
\begin{equation}
  \label{eq:Hh}
  \hat{H}_m = \frac{2}{n}\left\{\log L_n(\betabh) -  \log L_n(\betabh_0) \right\},
\end{equation}
where $ \log L_n(\betab) = \sum_{i=1}^n Y_i\X_i^T\betab - b(\X_i^T\betab) $ is the log-likelihood function and $\betabh$
and $\betabh_0$ are estimators of $\betab$ and $\betab_0$, respectively.

In the high-dimensional MGLM, we propose a two-step procedure to obtain $\betabh$ and $\betabh_0$. In the first step, we
employ a Sure Independence Screening \citep{Fan2010} method to screen the marginal effects of variables on the outcome and
threshold those effects to reduce the ultra-high dimensional problem to a moderate dimensional one. In the
second step, we solve a partially penalized maximum likelihood problem to obtain the final estimators.

In the screening step, define the Maximum Marginal Likelihood Estimator (MMLE) as
\begin{equation}
  \label{eq:8}
  \betabh_j^M = \argmin _{\betab_j} -\frac{1}{n}\sum_{i=1}^{n}l(Y_i, \Xt_{ij}^T\betab_j),
\end{equation}
for $j=1,\ldots,p$, where $l(Y_i,\Xt_{ij}^T\betab_j)= Y_i\Xt_{ij}^T\betab_j -b(\Xt_{ij}^T\betab_j)$,
$\betab_j = (\beta_{cj}, \beta_j)^T$, $\Xt_{ij} = (1,X_{ij})^T$, and $X_{ij}$ is the $j$-th element of $\X_i$.

The population counterpart of $\betabh_j^M$ is defined as
$\betab_j^M = \argmin _{\betab_j} \E  \{  -l(Y_1,\allowbreak \Xt_{1j}^T\betab_j) \}$. Then, we threshold $\betabh_j^M$ to obtain the set of
screened variables $\calMt=\{j: |\betah_j^M| \geq \gamma_n\}$, where $\gamma_n$ is a user-defined threshold level. We let
$\tilde{s}=|\calMt|$ be number of screened variables. As will be discussed in Section \ref{sec:statistical-results}, the
purpose of screening is to reduce the problem from an ultra-high dimensional one to a moderate dimensional one, while
ensuring that the screened set covers the true set, i.e., $\calMt\supset\mathcal{M}=\{j:\betas_j\neq 0 \}$, with high
probability. For that reason, we suggest setting a relatively small $\gamma_n$. In this paper, we use the following
procedure to choose the optimal $\gamma_n$. We conduct a grid search within a user-specified interval of thresholds. For
each grid point ${\gamma}^{\dagger}_n$, we regress the outcome $\Y$ on the subset of covariates
$\X_{\mathcal{M}^{\dagger}}$, where $\mathcal{M}^{\dagger}=\{j: |\betah_j^M|\geq \gamma_n^{\dagger} \}$, to obtain the
maximum likelihood estimator under that model. Then, we calculate the corresponding Bayesian Information Criterion (BIC)
value. It is important to note that the lower bound of the user-specified interval should be chosen such that
$|\mathcal{M}^{\dagger}|<n$, ensuring that the maximum likelihood estimators can be calculated. The optimal value of
$\gamma_n$ is chosen as the smaller one of the two thresholds that yields the smallest BIC value plus one standard error. We
remark that from a theoretical perspective, we only need to choose a threshold so that $\calMt \supset \mathcal{M}$ and the
size of $\calMt$ is controllable. For more details, see our discussion before Theorem \ref{thm:1}. In practice, some other
screening procedures can also be employed such as the Iterative Sure Independence Screening (ISIS) method \citep{Fan2010}.

In the second step, to estimate $\betab$, we propose to solve the following partially penalized maximization likelihood
problem:
\begin{equation}
  \label{eq:9}
  \betabh =  \argmin_{\betab_{\calMt^{\perp}}=\zero}  -\frac{1}{n} \sum_{i=1}^{n} \{Y_i\X_{i}^T\betab - b(\X_{i}^T\betab) \}
  + \sum_{j \in \calMt\backslash\calMt_m}p_{\lambda_1}(|\beta_j|), 
\end{equation}
where $\calMt_m = \{j: j \in \calMt \text{ and } j \in \text{the } m \text{-th modality}\}$ is the set of screened variables
in the $m$-th modality, $p_{\lambda_1}(|\beta_j|)$ is a folded-concave penalty function with a tuning parameter
$\lambda_1$. The folded concave penalty function $p_{\lambda}(|t|)$ is required to satisfy the following conditions: (i)
$p_{\lambda}(|t|)$ is non-decreasing and concave in $t\in \Rcal_{+}$ with $p_{\lambda}(0)=0$ and $p_{\lambda}(t)>0$ if
$t>0$; (ii) $p_{\lambda}(t)$ is differentiable for any $t\in \Rcal_+$; (iii) the first derivative $p'_{\lambda}(t)=0$ for
any $t\geq a\lambda$; (iv) $0\leq p'_{\lambda}(t)\leq \lambda$ for any $t\geq 0$. Examples of folded concave penalty
functions include the SCAD penalty \citep{fan2001variable} and the MCP penalty
\citep{zhang2010nearly}, 
among others. In (\ref{eq:9}), we only penalize the screened variables outside the $m$-th modality. Variables in
$\tilde{\calM}_m$ are not penalized in order to avoid introducing bias in their estimators. This has been shown to be
critical for correctly conducting corresponding statistical inference \citep{shi2019linear}. Moreover, as variables out of
$\tilde{\calM}$ have been screened out in the first step, $\betabh_{\calMt^{\perp}}=0$.

For the reduced model without the $m$-th modality, we solve another problem
\begin{equation}
  \label{eq:11}
  \betabh_0 = \argmin_{\substack{\betab_{\calMt^{\perp}}=\zero, \betab_{\calMt_m}=\zero}  }
  -\frac{1}{n} \sum_{i=1}^{n} [Y_i\X_{i}^T\betab -
  b(\X_{i}^T \betab)] +  \sum_{j \in\calMt\backslash \calMt_m}p_{\lambda_2}(|\beta_j|), 
\end{equation}
where $p_{\lambda_2}(|\beta_j|)$ is also a folded concave penalty function with a tuning parameter $\lambda_2$. In practice,
the same function can be used for $p_{\lambda_1}$ and $p_{\lambda_2}$, with the two parameters being separately tuned. We
utilize the SCAD penalty for both of them in our paper and choose optimal values of $\lambda_1$ and $\lambda_2$ that
minimize the BIC of the full and reduced models. In (\ref{eq:11}), we enforce all variables in the $m$-th modality to have
zero coefficients. Therefore, $\betabh_0=(\betabh^T_{0_{-m}},\zero^T)^T\in \Rcal^p$, where $\betabh_{0_{-m}}$ can be treated
as an estimator of $\betab_{0_{-m}}$ for the reduced model. Finally, we remark that the penalties used in (\ref{eq:9}) and
(\ref{eq:11}) further reduce the number of non-zero variables out of the $m$-th modality after the screening step. This
yields sparser estimators in the final stage and control the estimation bias. In Section \ref{sec:Num}, we will demonstrate that this extra regularization
also improves the accuracy of the statistical inference on $H_m$.

\section{Theoretical Results}
\label{sec:statistical-results}

In this section, we present two key statistical results on $\hat{H}_m$. We first prove that $\hat{H}_m$ is consistent to
$H_m$ and give the error bound. Then, we prove that the limiting distribution of $n\hat{H}_m$ follows a non-central
$\chi^2$-distribution, with a non-centrality parameter closely relates to $nH_m$. This distributional result enables us to
construct confidence intervals for $H_m$ based on $\hat{H}_m$, as well as computing p-values to assess the statistical
significance of a given modality. These findings meet the critical need on how to perform statistical inference at the
modality level for MGLM.

We first introduce some notations. For a vector $\a\in \mathcal{R}^d$, let $\supnorm{\a}=\max_j|a_j|$,
$\lonenorm{\a}=\sum_{j=1}^d |a_j|$, $\ltwonorm{\a}=(\sum_{j=1}^d a_j^2)^{1/2}$ denote its sup-norm, $L_1$-norm, and
Euclidean norm, respectively. For an index set $S$, let $\a_S$ denote the subvector of $\a$ with indices in $S$. In
particular, let the subscripts $m$ and $-m$ denote the sets of indices in and out of the $m$-th modality, respectively. Let
$|S|$ denote the size of $S$. For a square matrix $\A\in \mathcal{R}^{d\times d}$, let $\lambda_{\min}(\A)$ and
$\lambda_{\max}(\A)$ denote its minimal and maximal eigenvalues. Let $\supnorm{\A}=\sup_{ij}|a_{ij}|$,
$\lonenorm{\A}=\max_{1\leq j\leq d} \sum_{i=1}^d |a_{ij}|$, $\ltwonorm{\A}=\lambda_{\max}(\A)$,
$\fnorm{\A}=(\sum_{i,j} a_{ij}^2)^{1/2}$ denote its element-wise sup-norm, $L_1$-norm, $L_2$-norm, and Frobenious norm,
respectively. For index sets $S_1$ and $S_2$, let $\A_{S_1S_2}$ denote the submatrix of $\A$ with rows in $S_1$ and columns in
$S_2$. Finally, we introduce the
following qualities that will be used in the theoretical results. Recall the definitions of $\mathcal{M}$ and
$\tilde{\mathcal{M}}$ in Section \ref{sec:model}, let $s=|\mathcal{M}|$ and $\tilde{s}=|\tilde{\mathcal{M}}|$. Let
$\mathcal{M}_m $ and $\tilde{\mathcal{M}}_m$ be the subsets of $\mathcal{M}$ and $\tilde{\mathcal{M}}$ with indices in the
$m$-th modality. Define $s_m=|\mathcal{M}_m|$ and $\tilde{s}_m=|\tilde{\mathcal{M}}_m|$. To give the theoretical results, we
impose the following regularity conditions.

\con \label{cond1} $\x \in \Rcal^p$ is a sub-Gaussian random vector with parameter $\sigma^2$. That is, for any
$\mathbf{v} \in \Rcal^p$ with $\| \mathbf{v} \|_2 = 1$, $P(|\mathbf{v}^T\x| > t) \leq \exp \{-{t^2}/{(2\sigma^2)} \}$.
$c\leq \lambda_{\min}(\Sigmab)\leq \lambda_{\max}(\Sigmab) \leq C$, 
where $\Sigmab=\mathrm{Var}(\x)$, $c$ and $C$ are some generic positive constants.

\con \label{cond2} There exist generic positive constants $c$ and $C$ such that
$\E[\exp\{b(\x^T\betabs + c) \} - \exp\{b(\x^T\betabs) \}] + \E[\exp\{b(\x^T\betabs - c) \} -
\exp\{b(\x^T\betabs) \}] \leq C$.

\con \label{cond3} $\log p = O(n^a)$ for some $0 < a < 1$ and $\gamma_n = c_0 n^{-\kappa}$ for some $c_0>0$ and
$ 0 < \kappa < {(1-a)}/{2}$.

\con \label{cond4} The marginal Fisher information
$\I_j(\betab_j) = \E\{b''(\xt_j\betab_j) \xt_j\xt_j^T \}$ is finite and positive definite at 
$\betab_j = \betab_j^M$, where $\betab_j^M$ is defined below (\ref{eq:8}) and $\tilde{\x}_j=(1,x_j)^T$, $j = 1,\dots,p$. Moreover,
$\|\I_j(\betab_j)\|_2 < \infty$ for all
$\betab_j\in \mathcal{B}= \{\betab_j=(\beta_{cj},\beta_j)^T: |\beta_{cj}| \leq C, |\beta_j| \leq C \}$, where $C$ is a generic positive
constant.

\con \label{cond5} For $\betab_j \in \mathcal{B}$,
$\E\{l(y,\xt_j^T\betab_j)- l(y,\xt_j^T\betab_j^M) \} \geq c\|\betab_j-\betab_j^M \|_2^2$ for some generic positive
constant $c$, where the function $l(\cdot)$ is defined in (\ref{eq:8}).

\con \label{cond6} 
There exists a generic positive constant $C$, such that $|b''(\theta_1) - b''(\theta_2)| \leq C|\theta_1 - \theta_2|$ for
any $\theta_1$ and $\theta_2$. 

\con \label{cond7} There exists a positive constant $\delta$ such that, for some $0 < q < {1}/{2}-\kappa$,
$\sup_{\betab_j \in \mathcal{B}, \|\betab_j - \betab_j^M\|_2 \leq \delta} \E [ b(\tilde{\x}_j^T\betab_j) I(|x_j| >
Cn^{q})] = o(n^{-1})$.

\con \label{cond8} $\min_{j \in \calM^*} |\beta_j^M| \geq cn^{-\kappa}$ for some constants $c>0$ and
$0 < \kappa < {(1-a)}/{2}$.

\con \label{cond9} There exist some generic positive constants $c$ and $C$ such that
$c < \mathrm{Var} (\x^T \betabs ) < C$. Moreover, $\E[b''(\x^T\betabs)^{4}] < \infty$. 

\con \label{cond10} $\lambda_1 = c_1\sqrt{{(\log \tilde{s})}/{n}}$, $\lambda_2 = c_2\sqrt{{(\log \tilde{s})}/{n}}$ for some
positive constants $c_1, c_2$.

Condition~\ref{cond1} is a typical sub-Gaussian assumption for high-dimensional statistical problems. 
Condition~\ref{cond2}, which is analogous to the condition D from \cite{Fan2010}, 
ensures that the response variable has exponentially light tail. Condition~\ref{cond3} imposes assumptions on $p$ and the
threshold $\gamma_n$. Conditions \ref{cond4} and \ref{cond5} are assumptions about the marginal likelihood function, which
are analogous to the regularity conditions A' and C' imposed in \cite{Fan2010}.  Condition~\ref{cond6} ensures the smoothness of the
link function. \cite{Fan2010} 
proved that Conditions \ref{cond4} -- \ref{cond6} are satisfied for most common generalized
linear models.
Condition \ref{cond7} -- \ref{cond9} are some technical conditions needed in the proof. Condition~\ref{cond10}
gives the rate of a proper choice of tuning parameters.

Under Conditions \ref{cond1} -- \ref{cond8}, \cite{Fan2010} proved that the sure screening property holds with high
property, that is $P(\mathcal{M} \subset \hat{\mathcal{M}}) \rightarrow 1 $ as $n \rightarrow \infty$. Besides, they also
proved that $\tilde{s}=O_P\left(n^{2\kappa} \right)=o_P(n)$. These results show that the SIS can reach the sure screening
property with controllable size of screened variables. This is critical for us to establish the error bound of
$\hat{H}_m-H_m$ in below.

\begin{thm}
  \label{thm:1}
  Under Conditions~\ref{cond1} -- \ref{cond10}, it holds that
  \begin{equation*}
    \hat{H}_m - H_m = O_P(\max \{ {sn^{-1}\log n},{1}/{\sqrt{n}} \}).
  \end{equation*}
\end{thm}

Theorem 1 implies that $\hat{H}_m$ is consistent with $H_m$ under the sparsity assumption that $s=o(n/\log n)$ with a
convergence rate no faster than the typical root-n rate. Next, we derive the asymptotic distribution of $\hat{H}_m$. This
task is inherently a high-dimensional inference problem. The bulk of existing high-dimensional inference methods focus on
testing the significance of covariate coefficients in a high-dimensional regression model, or their linear combinations
\citep{javanmard2014confidence,Ning2017,shi2019linear,VarderGeer2014,zhang2014confidence,ZhuBradic2018}. However, as seen in
Section \ref{sec:model-setup}, $\hat{H}_m$ and $H_m$ both depend on some complex functions of the covariate coefficients and
other populations parameters. \cite{tony2020semisupervised} have developed some methods to infer the explained variance in a
high-dimensional linear model, where the explained variance is a quadratic function of the covariate coefficients, which
differs from our proposed metric. Moreover, their inference problem is in the context of a linear model. Hence, their
results are not applicable to our problem. Furthermore, unlike the hypothesis testing problems considered in previous works, which
focus more on the limiting distribution of the test statistic under the null hypothesis, we aim to derive the limiting
distribution of $\hat{H}_m$ for all cases of $H_m$. This distinction introduces additional complexity to our inference
task. Nevertheless, we find that we can still borrow ideas from existing works such as \cite{shi2019linear} to obtain the
limiting distribution of $\hat{H}_m$. Besides, since we adopt a two-step procedure in obtaining $\hat{H}_m$, the first
screening step can efficiently reduce the high-dimensional inference problem to a moderate-dimensional one and control the 
estimation bias introduced by high-dimensionality and simplify the statistical inference. We find that due
to the sure screening property, our asymptotic distributional result no longer needs to be built upon variable selection
consistency. As to be seen in conditions of Theorem \ref{thm:2}, we don't require any irrepresentable condition, which is
different from the conditions needed in \cite{shi2019linear}.
To obtain the asymptotic distribution of $\hat{H}_m$, we require the following additional regularity conditions:
\con \label{cond11} $ \|\betabs_{\calM_m}\|_2 = O(n^{\kappa-1/2})$. 

\con \label{cond12} $\max_{j} \sup_{\betab \in \mathcal{N}_0} \lambda_{\max}\{ \bX ^T \diag \{ |X_j| \circ b'''(\bX \betab)\} \X \} = O_P(n)$, and \\
$\sup_{\betab \in \mathcal{N}_0} \allowbreak \lambda_{\max}( \E\{b''(\x^T\betab) \x\x^T\})  < \infty$, where 
$\mathcal{N}_0 = \{ \betab \in \mathcal{R}^p : \| \betab - \betabs\|_2\ \leq C \sqrt{n^{-1} s \log n},~ \allowbreak \betab_{\mathcal{\Tilde{M}}^{\perp}}=\zero  \}$ for some constant $ C > 0$.

\con \label{cond13} The constant $\kappa$ in Conditions \ref{cond3}, \ref{cond7}, \ref{cond8}, and \ref{cond11} satisfies
$0< \kappa < {1}/{6}$.

Condition \ref{cond11} requires the magnitude of $\betabs_{\calM_m}$ to be controllable. Condition \ref{cond12} is a
technical condition needed in the proof. Its first component is analogous to Condition (A1) in \cite{shi2019linear} and has
been shown to hold with probability tending to 1 under the sub-Gaussian assumption on $\X$ in their paper. Condition
\ref{cond13} imposes a more stringent requirement on $\kappa$ compared to its counterpart in Theorem \ref{thm:1}. This is
expected, as deriving the asymptotic distribution typically requires stronger conditions than proving consistency, which has
been extensively discussed in the literature, e.g., in \cite{shi2019linear}.

\begin{thm}
  \label{thm:2}
  Under Conditions~\ref{cond1} -- \ref{cond13},
  \begin{equation}
    \label{eq:7}
    \sup_{x>0} |\textrm{P}( n\hat{H}_m \leq x) - \textrm{P}(\chi_{\Tilde{s}_m,\Gamma_n}^2 \leq x)| = o_P(1),
  \end{equation}
   where $\Gamma_n = n \betab_{\calMt_m}^{*T}\Omegab_{mm}^{-1} \betabs_{\calMt_m}$ and $\Omegab_{mm} = (\zero, \I_{\calMt_m}) [\E\{b''(\x^T\betabs)\x_{\calMt}\x_{\calMt}^T\}]^{-1} (\zero, \I_{\calMt_m})^T$.

\end{thm}

Theorem 2 shows that the limiting distribution of $n \hat{H}_m$ can be approximated by the non-central $\chi^2$-distribution
with $\tilde{s}_m$ degrees of freedom and a non-centrality parameter of $\Gamma_n$. Next, we establish the relationship between
$H_m$ and $\Gamma_n$ in the following proposition. 
\begin{pro}\label{pro3} Under Conditions~\ref{cond1}, \ref{cond6}, \ref{cond11} and \ref{cond12}, it holds that\\
  (a) $n H_m \geq \Gamma_n + o(n)$; (b) Furthermore, if $\|\betabs_{-m} - \betab_{0_{-m}}\|_2 =o(n^{\kappa-\frac{1}{2}})$,
  then $n H_m = \Gamma_n + o(n)$.
\end{pro}

Theorem \ref{thm:2} and Proposition \ref{pro3} together enable us to construct a confidence interval for $H_m$. The problem
essentially lies in constructing a confidence interval for the non-centrality parameter of the non-central
$\chi^2$-distribution based on a point estimator, a topic previously studied by \cite{KENT1995147}. 
Inspired by this work, we propose to use the pivotal quantity $F_{\tilde{s}_m,\theta}(n\hat{H}_m)$ to construct the
confidence interval, where $F_{\tilde{s}_m,\theta}(t)$ denotes the cumulative distribution function of the non-central
$\chi^2$-distribution with $\tilde{s}_m$ degrees of freedom and a non-centrality parameter of $\theta$. It can be
verified that $F_{\tilde{s}_m,\theta}(t)$ is a strictly decreasing function of $\theta$ for any fixed $t$, and
$F_{\tilde{s}_m,\theta}(n\hat{H}_m)$ asymptotically follows a uniform distribution in $[0,1]$ as $n\hat{H}_m$
asymptotically follows the distribution of $F_{\tilde{s}_m,\theta}$. Suppose we can find a real number $C_l$ such that
$F_{\tilde{s}_m,C_l}(n\hat{H}_m)=1-\alpha$. Then, setting $\theta=\Gamma_n+o(n)$, we have
$ P(C_l/n\leq H_m) \geq P(C_l\leq \theta) = P(F_{\tilde{s}_m,\theta}(n\hat{H}_m)\leq F_{\tilde{s}_m,C_l}(n\hat{H}_m))
=1-\alpha $, where Proposition \ref{pro3}(a) is used in the first inequality.  Thus, the $(1-\alpha)$ confidence lower
bound of $H_m$ is given by $[C_l/n,+\infty)$. With the additional assumption that
$\|\betabs_{-m} - \betab_{0_{-m}}\|_2 =o(n^{\kappa-\frac{1}{2}})$, following the same argument as above and applying
Proposition \ref{pro3}(b), the $(1-\alpha)$ confidence interval of $H_m$, up to an $o(1)$ term, is given by
$[C_l/n,C_u/n]$, where $C_l$ and $C_u$ satisfy $F_{\tilde{s}_m,C_l}(n\hat{H}_m)=1-\alpha/2$ and
$F_{\tilde{s}_m,C_u}(n\hat{H}_m)=\alpha/2$.  Since the form of the cumulative distribution function of the non-central
$\chi^2$-distribution is complicated, directly solving for $C_l$ and $C_u$ is computationally challenging. In practice,
we propose to use a bi-section method along the real line to find $C_l$ and $C_u$ such that
$F_{\tilde{s}_m,C_l}(n\hat{H}_m)\approx 1-\alpha/2$ and $F_{\tilde{s}_m,C_u}(n\hat{H}_m)\approx \alpha/2$. Moreover,
Theorem \ref{thm:2} also enables us to test if $H_m$ is zero. Under the null hypothesis that $H_m=0$, it can be shown
that $n\hat{H}_m$ asymptotically follows the standard $\chi^2_{\tilde{s}_m}$-distribution. Thus, the p-value is given by
$P(\chi^2_{\tilde{s}_m}>n\hat{H}_m)$.

Finally, we remark that unlike some high-dimensional inference methods \citep{shi2019linear}, Theorem \ref{thm:2} does not
require variable selection consistency for either $\betabh$ or $\betabh_0$, which can be seen from the absence of the
irrepresentable condition in Theorem \ref{thm:2}. The simulation results in Section \ref{sec:Num} further demonstrate this
point as our proposed method still gives good coverage probabilities, even when the estimators contain false positives or
false negatives.

In practice, quantifying the information gain of individual modalities in multimodal models as a percentage is often
useful. To address this need, we propose the following pseudo-$R^2$ metric: $R_m^2=1-\exp(-H_m)$, which can be estimated by
$\hat{R}_m^2=1-\exp(-\hat{H}_m)$. It can be easily shown that $R_m^2\in [0,1]$ as $H_m$ is non-negative. Given that the transformation function mapping $H_m$ to $R^2_m$ is strictly monotonically increasing, we can derive the confidence interval for $R^2_m$ by directly applying this transformation function to the confidence interval of $H_m$. We remark that the pseudo-$R^2$s introduced by
\cite{Hu2006} for unimodal logistic regression can be treated as special cases of our proposed
pseudo-$R^2$ metric.

\section{Simulation}
\label{sec:Num}

We conducted simulation studies to assess the empirical performance of the proposed metric on the evaluation of
information gain of modalities. In all simulations, we chose $n=300$ and $p=600$. We generated i.i.d samples of
$\{\X_i \}_{i=1}^n$ from $N_p(\zero,\Sigmab)$, where $\Sigmab=0.8\I_p+0.2\J_p$, and $\J_p\in \mathcal{R}^{p\times p}$ is
matrix with all entries equal to 1. We consider the following three models.

Model 1 (Linear regression): We set $M=3$, $p_1=p_2=p_3=200$,
  $\betabs_1=(-1.15, 1, 1.75,\allowbreak \zero)^T\in \mathcal{R}^{200}$,
  $\betabs_2=(-0.6, \allowbreak -0.8, 0.45, \zero)^T\in \mathcal{R}^{200}$,
  $\betabs_3=(-0.75, 0.8, \allowbreak -0.75, \zero)^T\in \mathcal{R}^{200}$, and
  $\betabs=\delta\cdot({\betab_1^{\ast}}^T,{\betabs_2}^T,\allowbreak {\betabs_3}^T)^T\in \mathcal{R}^{600}$, where
  $\delta\in \{0.6, 0.8, 1, \allowbreak 1.2, 1.6, 2 \}$ controls the magnitude of $\betabs$. We generated the response vector from the
  linear model that $\Y=\X\betabs+\epsilonb$, where $\epsilonb=(\epsilon_1,\ldots,\epsilon_n)^T$ is a vector of i.i.d errors
  from $N(0,1)$.
  
Model 2 (Logistic regression): We set $M=2$, $p_1=p_2=300$,
  $\betabs_1=(0.5, -1, -1.6, \allowbreak 0.9,\zero)^T\in \mathcal{R}^{300}$,
  $\betabs_2=(0.4, 0.8, -0.7, -1.4,\zero)^T\in \mathcal{R}^{300}$, and
  $\betabs=\delta\cdot({\betab_1^{\ast}}^T,{\betabs_2}^T)^T\in \mathcal{R}^{600}$, where
  $\delta\in \{1,1.2,1.4,1.7,2, 2.3 \}$ controls the magnitude of $\betabs$. We generated the response vector $\Y$ from the
  logistic regression model that $P(Y_i=1|\X_i)=\exp(\X_i^T\betabs)/\{1+\exp(\X_i^T\betabs) \}$.

Model 3 (Probit regression): We set $M=2$, $p_1=p_2=300$, and
  $\betabs_1=(0.5, 0.6, -0.7, \allowbreak -0.9,\zero)^T\in \mathcal{R}^{300}$,
  $\betabs_2=(0.4, 0.5, -0.6, -0.7,\zero)^T\in \mathcal{R}^{300}$.  We let
  $\betabs=\delta\cdot({\betab_1^{\ast}}^T,{\betabs_2}^T)^T\in \mathcal{R}^{600}$, where
  $\delta \in \{1, 1.1, 1.2, 1.3, 1.4, 1.5 \}$ controls the magnitude of $\betabs$. We generated the response vector $\Y$
  from the Probit regression model that $P(Y_i=1|\X_i)=\Phi(\X_i^T\betabs)$, where $\Phi(\cdot)$ is the cumulative
  distribution function of standard normal distribution.

  We involved three methods into comparison. The oracle method computed the maximum likelihood estimators $\betabh$ and
  $\betabh_0$ by regressing the response variable on the set of variables with truly nonzero coefficients, with and
  without the $m$-th modality, respectively. Our proposed method, denoted as SIS+SCAD, first performed SIS, followed by
  solving (\ref{eq:9}) and (\ref{eq:11}) with a SCAD penalty. The screening threshold for SIS was chosen via a grid
  search procedure as described in Section \ref{sec:estim-expect-relat}. The tuning parameters $\lambda_1$ and
  $\lambda_2$ were chosen to minimize the BIC criteria of the full and reduced models.  The third method, denoted as
  SIS+refit, first performed SIS, and then solved (\ref{eq:9}) and (\ref{eq:11}) with $\lambda_1=\lambda_2=0$. In other
  words, this approach refitted the model using only those variables that survived the SIS screening procedure. For each
  method, we plugged $\betabh$ and $\betabh_0$ into (\ref{eq:Hh}) to obtain the corresponding $\hat{H}_m$. Then, we used
  the procedure described after Theorem \ref{thm:2} to construct the two-sided confidence intervals. In particular, for
  the oracle method, we used quantiles of the non-central $\chi^2$-distribution with the oracle $s_m$ degrees of
  freedom. For the other two methods, we used quantiles of the non-central $\chi^2$-distribution with $\tilde{s}_m$
  degrees of freedom. To calculate $H_m$, we used (\ref{eq:example1}) for Model 1. For Models 2 and 3, we used the
  following procedure to calculate $H_m$. First, we independently generated another 10000 samples from these two
  models. We regressed the responses on variables outside the $m$-th modality with truly nonzero coefficients and
  obtained the maximum likelihood estimator as an approximation of $\betab_0$. Then, we plugged it and $\betabs$ into
  (\ref{eq:3}) and used the sample average over these 10000 samples as an approximation of $H_m$. For each model, we
  conducted 5000 simulation replications to evaluate the empirical coverage probabilities the two-sided confidence
  intervals for $H_m$. We also represented the sensitivity and specificity of
  variable selection performance given by SIS+SCAD and SIS+refit methods. 
  The results are represented in Figures \ref{f:lm} -- \ref{f:prob} for linear, logistic and probit regression models,
  respectively. 

Each one of the Figures \ref{f:lm} -- \ref{f:prob} shows that our method's confidence intervals exhibit satisfactory coverage probability when $H_m$ is not small.
In general, when $H_m$ increases, the coverage probability of our method also increases, which is expected. Comparing
coverage probabilities between the two modalities in each model shows that Modality 1 consistently achieves higher
coverage than Modality 2. This superior performance can be attributed to Modality 1's larger $H_m$ value, which further
validates our previous findings. 
We observe that when $H_m$ is relatively large, the coverage probability of our method
is close to that of the oracle method. Considering that our method deals with high-dimensional modalities, whereas the
oracle method uses only the low-dimensional true method, this result is promising. Furthermore, our method achieves
higher coverage probability compared to the SIS+refit approach, particularly when $H_m$ is small. This demonstrates that
the second-stage partially penalized method effectively eliminates false positives, leading to improved coverage
probability. This improvement is
evident in the specificity panels of Figures \ref{f:lm}, \ref{f:log}, and \ref{f:prob}, where our method consistently
achieves higher specificity than the SIS+refit method. Moreover, our method does not require a
perfect variable selection in order to yield good coverage probabilities. For example, in the Logistic and Probit
regressions, the sensitivity panels of Figures \ref{f:log} and \ref{f:prob} reveal that our method occasionally
omits some important variables. Nevertheless, it can still reach the nominal 95\% coverage probability when $H_m$ is
relatively large. This observation aligns with our discussion of the theoretical results following Theorem \ref{thm:2}.


\begin{figure}
  \centering
  \includegraphics[width=0.8\textwidth]{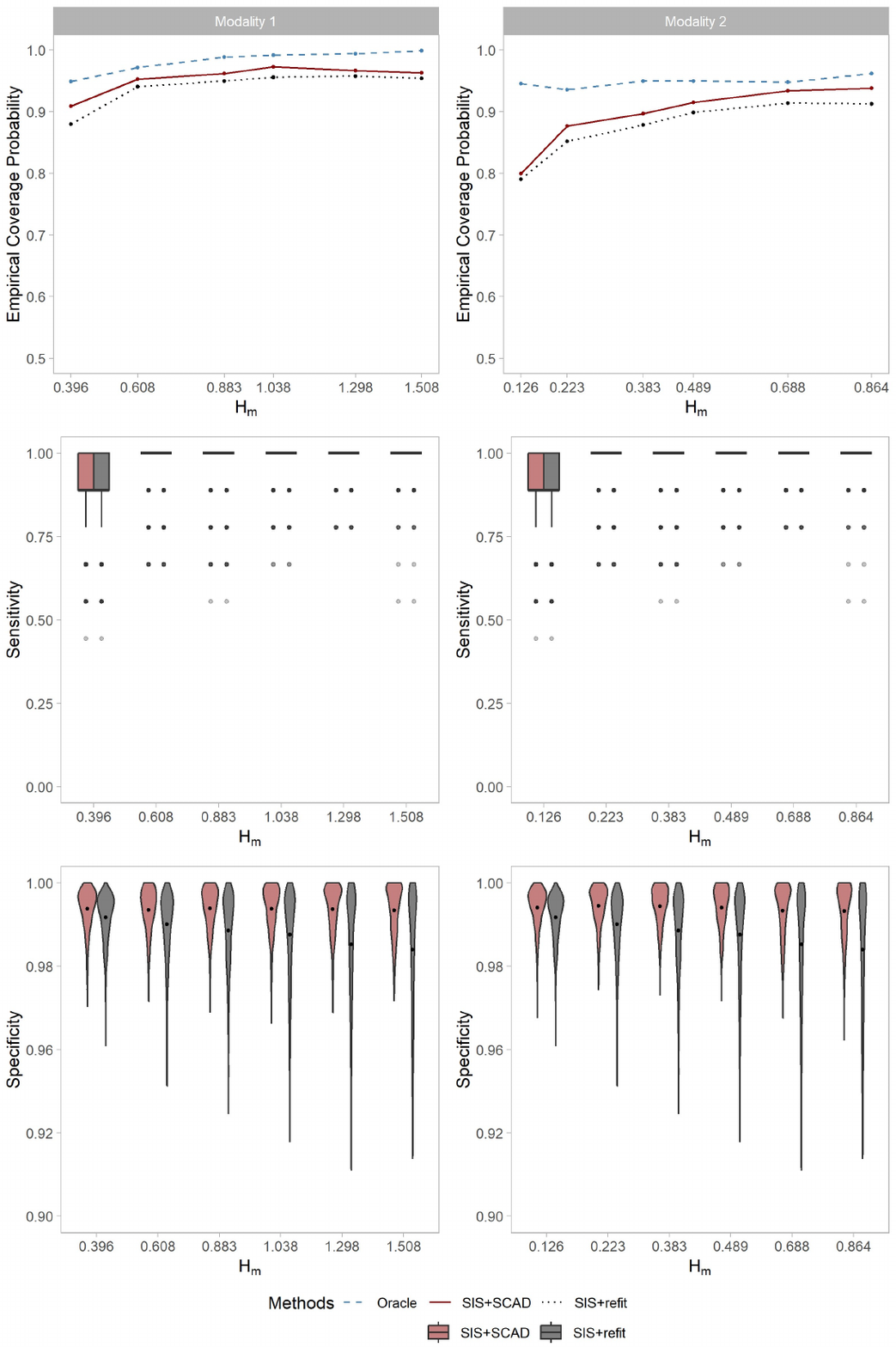}
  \caption{ {Empirical coverage probabilities and variable selection performance (Sensitivity and Specificity) by
      the three methods for Model 1. The oracle method (blue), SIS + SCAD penalized method (red), SIS + refit method
      (gray).}  }
  \label{f:lm}
\end{figure}

\begin{figure}
  \centering
  \includegraphics[width=0.8\textwidth]{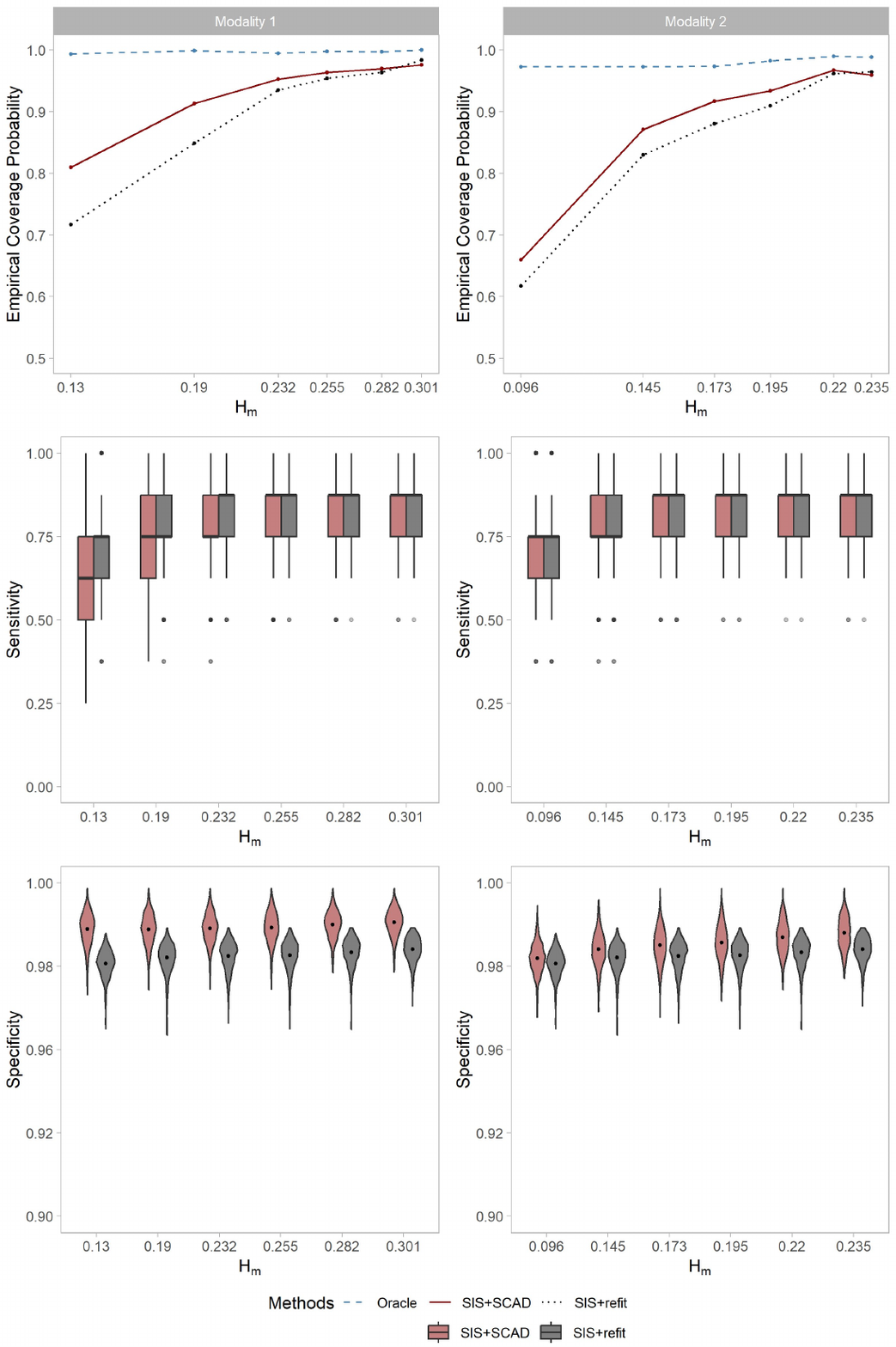}
  \caption{Empirical coverage probabilities and variable selection performance (Sensitivity and Specificity) by the three methods for Model 2. The oracle method (blue), SIS + SCAD penalized method (red), SIS + refit method (gray).}
  \label{f:log}
\end{figure}

\begin{figure}
  \centering
  \includegraphics[width=0.8\textwidth]{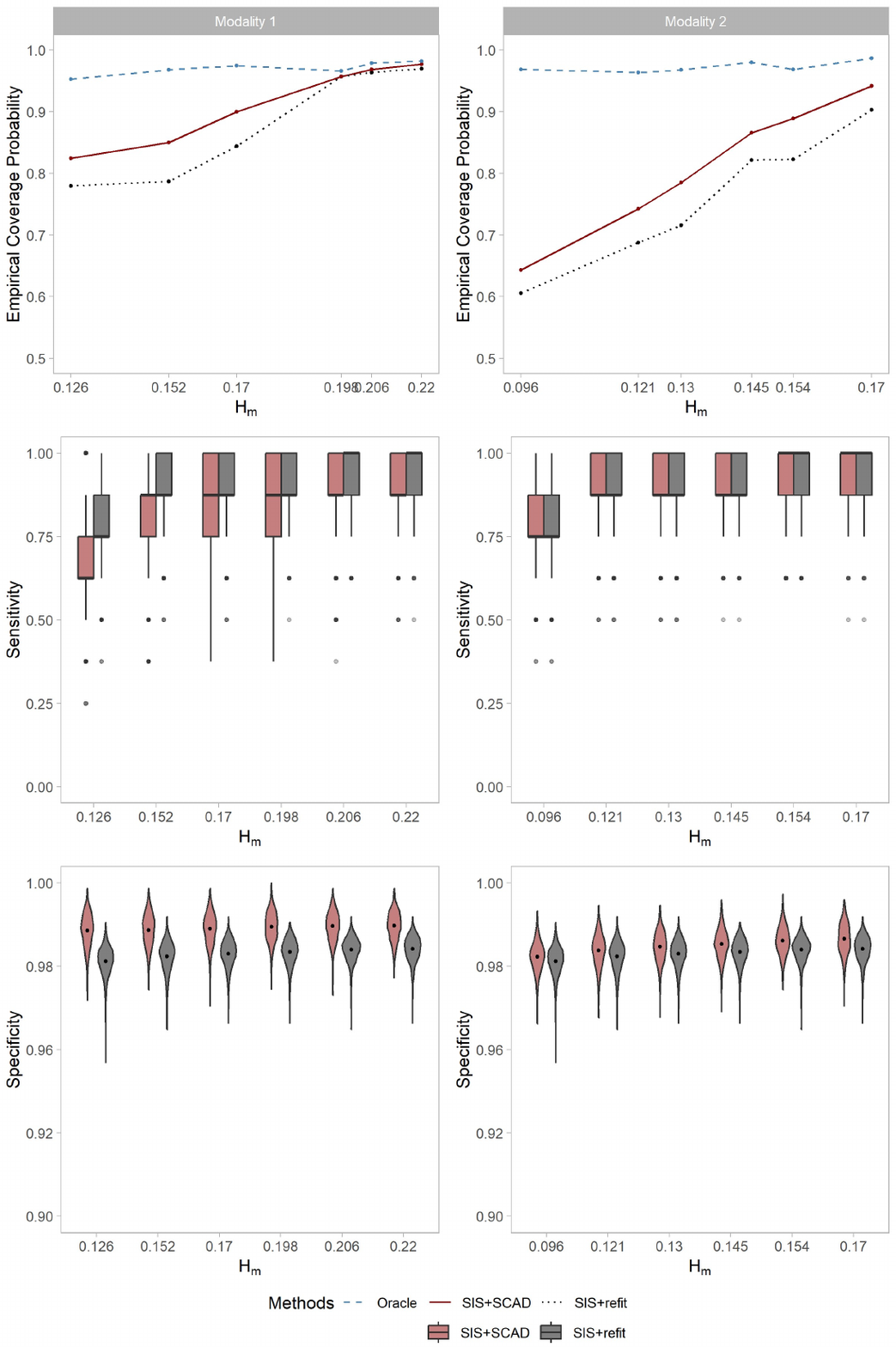}
  \caption{Empirical coverage probabilities and variable selection performance (Sensitivity and Specificity)  by the three methods for Model 3. The oracle method (blue), SIS + SCAD penalized method (red), SIS + refit method (gray).}
  \label{f:prob}
\end{figure}

\section{\label{sec:Num-neuro} Analysis of multimodal neuroimaging data}
We applied our proposed method to analyze a multimodal neuroimaging dataset from the Alzheimer's Disease Neuroimaging
Initiative (ADNI) database. 
Alzheimer's disease (AD) is an irreversible neurodegenerative disorder characterized by progressive impairment of
cognitive and memory functions. It is a major form of dementia for elderly subjects and is one of the leading causes of
death in the United States. Neuroimaging techniques, such as positron emission tomography (PET), provide valuable tools
for the diagnosis of AD. Our study included $n=669$ subjects with Amyloid-PET and FDG-PET scans from the ADNI
database. The Amyloid-PET image detects the presence and distribution of amyloid-beta proteins in the brain, which is a
crucial biomarker of AD. The FDG-PET image uses fluorodeoxyglucose (FDG) to measure metabolic activity and glucose
uptake in the brain, which provides insights into brain function and identifies areas of abnormal activity in
neurological disorders like AD. We used Destrieux atlas \citep{destrieux2010automatic} to parcellate each brain into 160
regions of interest (ROIs), consisting of 148 cortical regions and 12 sub-cortical regions $(p_1=p_2=160)$. Three
measurements were used to assess the cognitive decline of AD patients: the memory score (MEM), the executive function
score (EF), and the AD diagnostic label (DX). MEM and EF are continuous scores, while DX is an ordinal variable with
levels ranging from normal (NC) to early mild cognitive impairment (EMCI), stable MCI (SMCI), late MCI (LMCI), and AD.

In our analysis, we regressed MEM and EF on the two imaging modalities using the linear regression. Additionally, we
dichotomized DX, categorizing individuals with LMCI and AD as the case group and the remaining participants as the control
group. We then regressed the dichotomized DX variable on the two imaging modalities using the logistic regression. For each
model, we first performed SIS followed by solving the partially penalized problem with the SCAD penalty in (\ref{eq:9}) to
select ROIs that associated with the outcome. The selected ROIs are shown in Figure \ref{f:ANDI}. The results indicate that
the MEM score is more sensitive to the accumulation of AD-related biomarkers than the EF score, as more ROIs are selected in
MEM-based model than EF-based model. This is aligned with current clinical findings that memory decline is a typical
clinical symptom \citep{jahn2013memory}. Moreover, we used our proposed method to calculate the estimated ERE with 95\% confidence intervals and
corresponding p-values for each of the two modalities and regression models. The results, shown in Table \ref{tab:1},
suggest that FDG-PET is significantly more informative than Amyloid-PET for the outcomes of EF and DX. This shows that the
metabolism level is a more putative biomarker for monitoring cognitive decline than pathology biomarkers such as
amyloid. 

\begin{figure}
  \centering
  \includegraphics[width=5in]{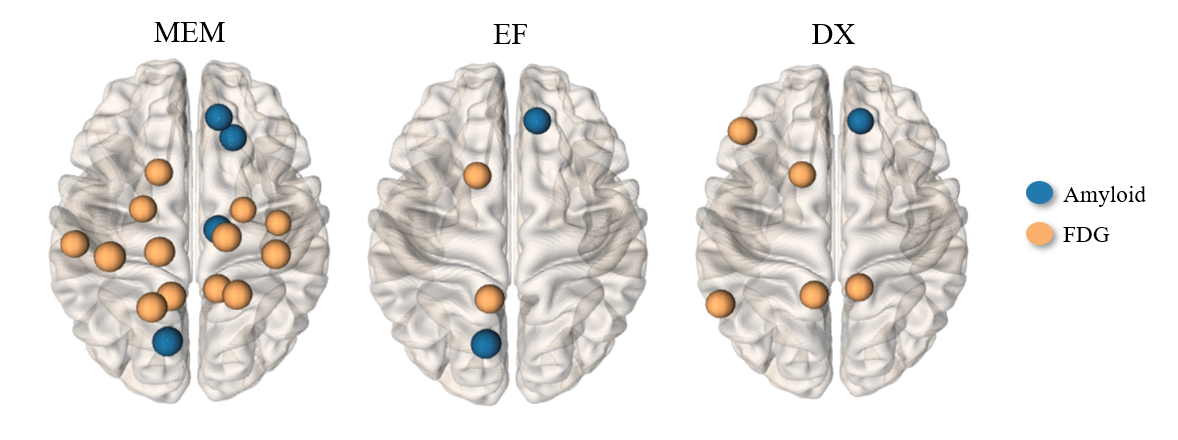}
  \caption{Selected regions of interest for MEM, EF and DX by two modalities. 
  MEM, the memory score; EF, the executive function score; DX, the AD diagnostic label.}
  \label{f:ANDI}
\end{figure}

\begin{table}
\centering
\caption{Point estimates, confidence intervals and p-values of the Expected Relative Entropy for different Neuroimaging techniques. (MEM: the memory score; EF: the executive function score; DX: the AD diagnostic label.)}
{\begin{tabular}{@{\extracolsep{4pt}}l*{2}{l@{\enspace}l@{\enspace}l}}
  \hline\hline
 {}  & \multicolumn{3}{c}{Amyloid-PET} & \multicolumn{3}{c}{fluorodeoxyglucose-PET}\\
 \cmidrule{2-4}
 \cmidrule(lr){5-7}
 Outcome & $\hat{H}_m$ & 95\% CI & p-value & $\hat{H}_m$ & 95\% CI & p-value\\
   \midrule
   MEM & 0.170 & (0.109, 0.233) & $1.19\times 10^{-23}$ & 0.147 & (0.078, 0.191) & $3.93 \times 10^{-15}$ \\
  EF  & 0.089 & (0.048, 0.139) & $1.10 \times 10^{-13}$ & 0.230 & (0.161, 0.306) & $4.28 \times 10^{-34}$\\
  DX & 0.033 & (0.011, 0.066) & $2.77 \times 10^{-6}$ & 0.168 & (0.106, 0.229) & $1.5\times 10^{-22}$\\
 \hline\hline
\end{tabular}}
  \label{tab:1}
\end{table}

\section{ Discussion}
\label{sec:discuss}

We propose an ERE metric to quantify the information gain of individual modalities in high-dimensional multimodal
generalized linear models. The metric is justified through some discussion of its mathematical properties and is related to
conventional goodness-of-fit metrics for unimodal models, as well as recent developments in modality evaluation under more
restrictive models. We demonstrate that these metrics can be viewed as special cases or are equivalent to our proposed
metric. Moreover, we develop a deviance-based estimator of the ERE metric by using a two-step procedure. We derive the error
bound and the asymptotic distribution of this estimator, which enables us to perform statistical inference on the ERE. This
meets the growing demand for rigorous statistical inference tools at the modality level in multimodal models.  Through
theoretical studies and numerical experiments. We demonstrate that our constructed confidence intervals maintain good
coverage probabilities without relying on variable selection consistency, a desirable property given the difficulty of
achieving perfect variable selection in high-dimensional models.

While our study focuses on generalized linear models, we note that our entropy-based metric can possibly be extended to
other models, such as the Cox proportional hazard model. Further research into such extensions would be valuable, as it
would provide a comprehensive tool set for assessing the significance of individual modalities across various multimodal
models.

\section{Acknowledgment}
Dr. Li's research is partially supported by NIH grant R01-AG073259. 

\bibliographystyle{biom}
\bibliography{reference}

\newpage

\begin{center}
  \textbf{\Large Supplements Materials}
\end{center}

\setcounter{equation}{0}
\setcounter{section}{0}
\renewcommand{\thesection}{S\arabic{section}}

\renewcommand{\theequation}{S\arabic{equation}}
\renewcommand{\thesubsection}{\thesection.\arabic{subsection}}
\renewcommand{\thefigure}{S\arabic{figure}}
\renewcommand{\thetable}{S\arabic{table}}
\renewcommand{\thepro}{S\arabic{pro}}

\section{Proofs}
\subsection{Proof of Proposition~\ref{pro1}:}
  Under the assumption that $\E|\log {p(y|\x_{-m})}/{p(y|\x_{-(m,l)})}| < \infty$, we have
  \begin{align*}
    (H_{(m,l)} - H_m)/2
    &= \E_{\x} [D_{\text{KL}}(p(y|\x)\parallel p(y|\x_{-(m,l)})) - D_{\text{KL}}(p(y|\x)\parallel p(y|\x_{-m}))] \\
    &= \int_{\x} p(\x)\int_{y}p(y|\x)\log \frac{p(y|\x_{-m})}{p(y|\x_{-(m,l)})}dy d\x \\
    &= \int_{\x_{-m}} p(\x_{-m}) \left[ \int_{y} \log \frac{p(y|\x_{-m})}{p(y|\x_{-(m,l)})} \int_{\x_{m}}p(y|\x)
      p(\x_m|\x_{-m}) d\x_m dy \right] d\x_{-m} \\ 
    &= \int_{\x_{-m}} p(\x_{-m}) \left[ \int_{y} p(y|\x_{-m}) \log \frac{p(y|\x_{-m})}{p(y|\x_{-(m,l)})} dy \right]  \\
    &= \int_{\x_{-m}} p(\x_{-m}) D_{\text{KL}}[p(y|\x_{-m})\parallel p(y|\x_{(-(m,l))})] d\x_{-m}
    \geq 0,
  \end{align*}
  where $p(\x)$ and $p(\x_{-m})$ denote the densities of $\x$ and $\x_{-m}$ respectively, and the last inequality follows
  from the non-negativity of the KL divergence.

\subsection{Proof of equation~\eqref{eq:example1}:} For the reduced model without the $m$-th modality,
  $y|\x_{-m} \sim N(\x_{-m}^T \betab_{-m}^*, \sigma^2_{\epsilon} + \sigma^2_{m|-m})$.  By Lemma 1 of
  \cite{belov2011distributions}, the KL divergence can be expressed as
  \begin{equation*}
    D_{KL}(p(y|\x) \| p(y|\x_{-m})) =  \frac{1}{2} \left\{
      \log\frac{\sigma_{m|-m}^2 + \sigma^2_{\epsilon} }{\sigma^2_{\epsilon}} +
      \frac{\sigma^2_{\epsilon}}{\sigma_{m|-m}^2 + \sigma^2_{\epsilon}} + 
      \frac{(\x^T\betabs - \x_{-m}^T\betab_{-m}^*)^2 }{\sigma_{m|-m}^2 + \sigma^2_{\epsilon}} -1 \right\},
  \end{equation*}
  where
  $\sigma_{m|-m}^2 = \betab_m^{*T} \Sigmab_{m|-m} \betab_m^* = \betab_m^{*T} (\Sigmab_m - \Sigmab_{m,-m}
  \Sigmab_{-m}^{-1}\Sigmab_{-m,m})\betab_m^*$, $\Sigmab_m = \E(\x_m \x_m^T)$,
  $\Sigmab_{-m} = \E(\x_{-m} \x_{-m}^T)$, $\Sigmab_{-m,m}= \E(\x_{-m} \x_m^T)$,
  $\Sigmab_{m,-m}= \E(\x_m \x_{-m}^T)$, and
  $\Sigmab_{m|-m}=\Sigmab_m - \Sigmab_{m,-m} \Sigmab_{-m}^{-1}\Sigmab_{-m,m}$. Then,
  \begin{align*}
    H_m = 2\E_{\x} [D_{KL}(p(y|\x) \| p(y|\x_{-m}))] 
        =  \log\frac{\sigma_{m|-m}^2 + \sigma^2_{\epsilon} }{\sigma^2_{\epsilon}} +
          \frac{ \betab_m^{*T} \Sigmab_m \betab_m^* - \sigma_{m|-m}^2}{\sigma_{m|-m}^2 + \sigma^2_{\epsilon}}.
  \end{align*}


\subsection{Proof of Theorem 1:} We have
  \begin{align*}
    \frac{1}{2} \hat{H}_m &= \frac{1}{n} \log L_n(\betabh) - \frac{1}{n} \log L_n(\betabh_0) \\
                        &= \frac{1}{n}\{\log L_n(\betabh) - \log L_n(\betabs) \} - 
                          \frac{1}{n}\{\log L_n(\betabh_0) - \log L_n(\betab_0 ) \}  +  \frac{1}{n}\{\log L_n(\betabs) - \log L_n(\betab_0 ) \}.
  \end{align*}
  First, we show that $n^{-1}\{\log L_n(\betabh) - \log L_n(\betabs) \} = O_P(sn^{-1} \log n)$ and
  $n^{-1} \allowbreak \{\log L_n(\betabh_0) - \log L_n(\betab_0 ) \} = O_P(sn^{-1} \log n) $. 
   Let
  \begin{align*}
    \S_n(\betab) &= -\frac{\partial \{n^{-1}\log L_n(\betab) \}}{\partial \betab} 
                   = -\frac{1}{n} \sum_{i=1}^{n}\{Y_i - b'(\X_i^T\betab) \}\X_i , \\
    \J_n(\betab) &= \frac{\partial \S_n(\betab)}{\partial \betab} 
                   = \frac{1}{n} \sum_{i=1}^{n}b''(\X_i^T\betab)\X_i \X_i^T .
  \end{align*}
  By the second-order Taylor's expansion,
\begin{equation}
\label{eq:Taylor}
    \frac{1}{n}\log L_n(\betabh) - \frac{1}{n}\log L_n(\betabs) 
    = \underbrace{-\S_n(\betabh)^T(\betabh-\betabs)}_{I}  
     +\underbrace{\frac{1}{2}(\betabh-\betabs)^T \J_n(\tilde{\betab})(\betabh-\betabs)}_{II},
\end{equation}
  where $\Tilde{\betab} = t \betabs + (1-t) \betabh$ for some $0 \leq t \leq 1$ is a vector between $\betabs$ and $\betabh$.

  For $I$, $\S_n(\betabh)^T(\betabh-\betabs) = \sum_{j=1}^{p} S_n(\betabh)_j (\betah_j-\beta^*_j) $.
  For $j \notin \calMt$, $\betah_j = 0$. Then, by the sure screening property, $\betah_j-\beta^*_j = 0$.
  For $j \in \calMt$, by the KKT condition of \eqref{eq:9}, 
  $$ 0 \in \frac{\partial Q_n(\betab)}{\partial \beta_j}\biggr|_{\beta_j = \hat{\beta}_j} = -S_n(\betabh)_j + I(j \notin \calMt_m) p'_{\lambda}( | \hat{\beta_j} |) ,$$
  where $Q_n(\betab) = -{n}^{-1}\log L_n(\betab) + p_{\lambda_1}(\betab_{\calMt_{-m}})$, $S_n(\betabh)_j$ denotes the $j$-th
  element of $\S_n(\betabh)$, and $p'_{\lambda_1}(\cdot)$ denotes the first derivative of $p_{\lambda_1} (\cdot)$. By the
  property of folded concave penalty function, we have $ |S_n(\betabh)_j| = |I(j \notin \calMt_m)p'_{\lambda_1}(\hat{\beta_j})| \leq \lambda_1 $.
  Therefore,
  \begin{equation}
    \label{eq:10}
    | \S_n(\betabh)^T(\betabh-\betab) | \leq 
    \sup_{j \in \calMt} | S_n(\betabh)_j | \cdot \| \betabh_{\calMt} - \betab_{\calMt}^* \|_1 \leq
    \lambda_1 \| \betabh_{\calMt} - \betab_{\calMt}^* \|_1.
  \end{equation}

  For II, we have
  \begin{align*}
    (\betabh-\betabs)^T \J_n(\tilde{\betab})(\betabh-\betabs) 
     = (\betabh_{\calMt} - \betab_{\calMt}^*)^T \J_n(\tilde{\betab})_{\calMt} (\betabh_{\calMt} - \betab_{\calMt}^*) ,
\end{align*}
where $\J_n(\tilde{\betab})_{\calMt} = {n}^{-1} \sum_{i=1}^{n}b''(\X_i^T\tilde{\betab}) \X_{i,\calMt} \X_{i,\calMt}^T $.
Under Conditions~\ref{cond2} and \ref{cond4}--\ref{cond8}, the conditions $A'$ -- $E$ for Theorems 4 and 5 in \cite{Fan2010}
are satisfied. Then, these two theorems imply that 
\begin{align*}
    P(\mathcal{M}_* \subset \calMt) &\geq 1 - s \{ \exp(-k_1 n^{1-2\kappa-q}) + n k_2\exp(-n^{2q})\}, \\
    P(|\calMt| = O( n^{2\kappa})) &\geq 1 - s \{ \exp(-k_1 n^{1-2\kappa-q}) + n k_2\exp(-n^{2q})\},
\end{align*}
where $k_1, k_2$ are some positive constants. Under Condition~\ref{cond7}, $1 - 2\kappa - 2q > 0 $. Thus, we have
$P(\mathcal{M}_* \subset \calMt) \rightarrow 1$ and $\tilde{s}=O_P(n^{2\kappa})=o_P(n)$. 
Then, it follows from Lemma \ref{lemma3} that
\begin{equation}
  \label{eq:12}
    \lambda_{\max} (\J_n(\tilde{\betab})_{\calMt} ) = O_P(1).
\end{equation}
Combining (\ref{eq:Taylor}), (\ref{eq:10}) and (\ref{eq:12}), we have
\begin{equation}
  \label{eq:13}
    \left| \frac{1}{n}\log L_n(\betabh) - \frac{1}{n}\log L_n(\betabs) \right| 
    \leq  \lambda_1 \| \betabh_{\calMt} - \betab_{\calMt}^* \|_1 + 
  C_3 \| \betabh_{\calMt} - \betab_{\calMt}^* \|_2^2,
\end{equation}
where $C_3$ is a generic positive constant. Under Condition \ref{cond1}, \cite{negahban2012unified} shows that the
negative log-likelihood function of the generalized linear model satisfies the Restricted Strong Convexity (RSC)
condition. 
Then, applying Theorem 1 and Corollary 1 of \cite{negahban2012unified}, we have,
\begin{equation}
  \label{Negahban}
  \begin{split}
      \| \betabh_{\calMt} - \betab_{\calMt}^* \|_1  = O_P \left( s \sqrt{\frac{ \log \tilde{s} }{n}} \right) 
      \text{ and } \| \betabh_{\calMt} - \betab_{\calMt}^* \|_2  = O_P \left( \sqrt{\frac{ s \log \tilde{s}}{n} } \right).
  \end{split}
\end{equation}
Combining (\ref{eq:13}) and (\ref{Negahban}) with $\lambda_1 = c_1 \sqrt{(\log \tilde{s})/{n}}$  and $\tilde{s} = O_P(n^{2\kappa})$, we have
\begin{equation}
    \label{eq:lm1.1}
    \frac{1}{n}\log L_n(\betabh) - \frac{1}{n}\log L_n(\betabs) = O_P \left(\frac{s \log \tilde{s}}{n} \right) = O_P \left(\frac{s \log n}{n} \right).
\end{equation}
Similarly, for $\betabh_0$ and $\betab_0$, we have
\begin{equation}
    \label{eq:lm1.2}
    \frac{1}{n}\log L_n(\betabh_0) - \frac{1}{n}\log L_n(\betab_0) = O_P \left(\frac{s \log \tilde{s}}{n} \right) =
    O_P \left(\frac{s \log n}{n} \right). 
\end{equation}

Next, we show that $2 n^{-1}\{\log L_n(\betabs) - \log L_n(\betab_0 ) \} = H_m + O_P({1}/{\sqrt{n}})$.
Since
\begin{equation}
    \label{eq:lm3.1}
    \log L_n(\betabs) - \log L_n(\betab_0 ) 
    = \sum_{i=1}^{n} \log p(Y_i|\X_{i}) - \sum_{i=1}^{n} \log p(Y_i|\X_{i,-m}),
\end{equation}
and
\begin{equation}
    \label{eq:lm3.2}
    H_m = 2 \textrm{E}_{\x} \Bigl\{ \E_{y|\x} \left[ \log \frac{p(y|\x)}{p(y|\x_{-m})} \right] \Bigr\}
    = 2 \textrm{E} \left[\log \frac{p(y|\x)}{p(y|\x_{-m})} \right],
\end{equation}
by Chebyshev's inequality, 
\begin{align*}
    P \left( \left| \frac{1}{n}\sum_{i=1}^{n} \log \frac{p(Y_i|\X_{i})}{p(Y_i|\X_{i,-m})} - \E \left\{\log \frac{p(y|\x)}{p(y|\x_{-m})} \right\} \right| > \epsilon \right) 
     \leq  \frac{\mathrm{Var}\left(\log \frac{p(y|\x)}{p(y|\x_{-m})}\right)}{n\epsilon^2}.
\end{align*}
Let $\Deltab = \betabs - \betab_0$. We have,
\begin{align*}
    \mathrm{Var} \left(\log \frac{p(y|\x)}{p(y|\x_{-m})} \biggr| \x \right) &= \mathrm{Var}(y\x^T\betabs - y\x^T\betab_0 - b(\x^T\betabs) + b(\x^T\betab_0) | \x ) \\
    & = \mathrm{Var}(y \x^T \Deltab | \x) = (\x^T\Deltab)^2 b''(\x^T \betabs).
\end{align*}
By Taylor's expansion of $b(\x^T\betab_0)$ around $\x^T\betabs$, we have 
\begin{equation} \label{eq: Taylor_bfun}
    b(\x^T\betab_0) = b(\x^T\betabs) - b'(\x^T\betabs) \x^T \Deltab + \frac{1}{2}
    b''(\x^T\tilde{\betab})(\x^T \Deltab)^2, 
\end{equation}
where $\tilde{\betab} = t \betab + (1-t) \betab_0$ for some $t\in (0,1)$. Then,
\begin{align*}
  \E \left[ \log \frac{p(y|\x)}{p(y|\x_{-m})} \biggr| \x \right]
  &= b'(\x^T\betabs) \x^T \betabs- b'(\x^T\betabs) \x^T \betab_0  -b(\x^T\betabs) 
   + b(\x^T\betab_0) \\
  &= \frac{1}{2} b''(\x^T\tilde{\betab}) (\x^T \Deltab)^2 ,\\
  \mathrm{Var}\{b''(\x^T\tilde{\betab}) (\x^T \Deltab)^2\}
  &\leq \E\{b''(\x^T\tilde{\betab})^2 (\x^T \Deltab)^4 \} 
   \leq (\E[b''(\x^T\tilde{\betab})^4])^{1/2} (\E[(\x^T \Deltab)^8])^{1/2}.
\end{align*}
Then, we have
\begin{equation}
  \label{eq:14}
  \begin{split}
  & \hspace{3ex}\mathrm{Var} \left(\log \frac{p(y|\x)}{p(y|\x_{-m})} \right) 
  = \E\left[\mathrm{Var}\left\{\log \frac{p(y|\x)}{p(y|\x_{-m})} \biggr| \x\right \}\right] +
    \mathrm{Var}\left[\E\left\{\log \frac{p(y|\x)}{p(y|\x_{-m})} \biggr| \x \right\} \right] \\ 
  & \leq \E\{(\x^T \Deltab)^2 b''(\x^T \betabs ) \}
    + \frac{1}{4} \E\{b''(\x^T\tilde{\betab})^2 (\x^T \Deltab)^4 \} \\
  & \leq (\E[b''(\x^T \betabs)^2])^{1/2} (\E[(\x^T \Deltab)^4])^{1/2} 
    +
    \frac{1}{4}(\E[b''(\x^T\tilde{\betab})^4])^{1/2} (\E[(\x^T \Deltab)^8])^{1/2} . 
  \end{split}
\end{equation}
By Condition \ref{cond1}, $\x^T \Deltab/\ltwonorm{\Deltab}$ is sub-Gaussian. 
Then, for $k =4 $ and $8$, $\E[(\x^T \Deltab)^k]< \infty$.
Under Conditions \ref{cond6} and \ref{cond9}, (\ref{eq:14}) implies that there exists a finite constant $C$ such
that $\mathrm{Var} \left(\log \frac{p(y|\x)}{p(y|\x_{-m})} \right) < C$. Taking $\epsilon = {D}/{\sqrt{n}}$ for some
sufficiently large constant $D$, we have
\begin{align*}
  P \left( \left |\frac{1}{n}\sum_{i=1}^{n} \log \frac{p(Y_i|\X_i)}{p(Y_i|\X_{i,-m})} - \E \left[ \log
  \frac{p(y|\x)}{p(y|\x_{-m})} \right] \right| > \frac{D}{\sqrt{n}} \right)  \leq \frac{C}{D}.
\end{align*}
This together with \eqref{eq:lm3.1} and \eqref{eq:lm3.2} imply that
\begin{equation}
    \label{eq:lm1.3}
    \frac{2}{n}\log L_n(\betabs) - \frac{2}{n}\log L_n(\betab_0 ) 
    = H_m + O_P({1}/{\sqrt{n}}).
\end{equation}
Combining \eqref{eq:lm1.1}, \eqref{eq:lm1.2} and \eqref{eq:lm1.3}, we have
\begin{equation}
    \hat{H}_m - H_m=  
    O_P\left(\max \left\{ \frac{s \log n}{n} ,\frac{1}{\sqrt{n}} \right\}\right).
\end{equation}

\subsection{Proof of Theorem 2}
  For $\hat{H}_m$, By Taylor's expansion, we have 
\begin{equation}
    \label{eq:Taylor_Hm}
 \begin{split}
   &\hspace{3ex} \frac{1}{2} \hat{H}_m = \frac{1}{n}\log L_n(\betabh) - \frac{1}{n}\log L_n(\betabh_0) \\
     &= \underbrace{\frac{1}{n}(\betabh - \betabh_0)^T \X^T \{\Y-\boldsymbol{b}'(\X \betabh) \}}_{I}
     + \underbrace{\frac{1}{2} (\betabh - \betabh_0)^T \J_n(\betabh) (\betabh - \betabh_0)}_{II} 
     + \underbrace{\frac{1}{6}(\betabh- \betabh_0)^T \bR_n }_{III},
\end{split}
\end{equation}
where $\R_n = (R_{n,1}, \dots, R_{n,p})^T$ with
$R_{n,j} = n^{-1}(\betabh - \betabh_0)^T\bX ^T \diag \{ |\X_j| \circ b'''(\bX \Tilde{\betab})\} \bX (\betabh - \betabh_0)$, and
$\Tilde{\betab} = t\betabh_0 + (1-t)\betabh$ for some $t\in (0,1)$.

First, we can show that  $I =o_P(1)$. Since $\mathcal{M}\subset \tilde{\mathcal{M}}$, 
\begin{equation}
    \label{eq:Thm2.1}
    \frac{1}{n}(\betabh - \betabh_0)^T \X \{\Y-\boldsymbol{b}'(\X\betabh) \} 
    = -(\betabh - \betabh_0)^T \S_n(\betabh) 
    =  -\sum_{j \in \calMt}
    (\hat{\beta}_j - \hat{\beta}_{0_j}) S_n(\betabh)_j. 
\end{equation}

We first derive the upper bound for $\lonenorm{\betabh - \betabh_0}$.   
It follows from Lemma \ref{lemma5} that  
\begin{equation}\label{eq:diff_bh}
    \begin{split}
        \sqrt{n} (\betabh_{\calMt}  - \betabh_{0_{\calMt}}) 
        &= \frac{1}{\sqrt{n}}  \J_{n_{\calMt}}^{-1/2} \P_m \J_{n_{\calMt}}^{-1/2} \X_{\calMt} ^T \{\Y-\boldsymbol{b}'(\X \betabs) \} 
        + \sqrt{n} \J_{n_{\calMt}}^{-1/2} \P_m \J_{n_{\calMt}}^{1/2}\Q_m\betabs_{\calMt_m} + o_P(1), 
        \end{split}
\end{equation} 
where $\J_{n_{\calMt}}=\J_n(\betabs)_{\calMt}$, $\Q_m=(\zero, \I_{\calMt_m})^T \in \mathcal{R}^{\tilde{s} \times \tilde{s}_m}$ and
$\I_{\calMt_m}$ is the $\tilde{s}_m$-dim identity matrix. 
From Lemma~\ref{lemma3}, $\lambda_{\max}(\J_{n_{\calMt}}) = O_P(1)$. Since
\begin{equation*}
    \E \ltwonorm{\frac{1}{\sqrt{n}} \P_m \J_{n_{\calMt}}^{-1/2} \X_{\calMt} ^T \{\Y-\boldsymbol{b}'(\X \betabs) \}}^2 = \text{tr}(\P_m) = \text{rank}(\P_m) = \tilde{s}_m,
\end{equation*}
it follows from Markov's inequality that
\begin{equation}\label{eq:diff_bh_1}
    \begin{split}
      & \hspace{3ex} \ltwonorm{ \frac{1}{\sqrt{n}}  \J_{n_{\calMt}}^{-1/2} \P_m \J_{n_{\calMt}}^{-1/2} \X_{\calMt} ^T \{\Y-\boldsymbol{b}'(\X \betabs) \} } \\
      &\leq \lambda_{\max}(\J_{n_{\calMt}}^{-1/2}) \ltwonorm{ \frac{1}{\sqrt{n}}\P_m \J_{n_{\calMt}}^{-1/2} \X_{\calMt} ^T
        \{\Y-\boldsymbol{b}'(\X \betabs) \} } = O_P(\sqrt{\tilde{s}_m}).
    \end{split}
\end{equation}
By Condition~\ref{cond11}, $\ltwonorm{\betabs_{\calMt_m}} = \ltwonorm{\betabs_{\calM_m}}= O(n^{\frac{2\kappa-1}{2}})$, then,
\begin{equation}\label{eq:diff_bh_2}
    \ltwonorm{\sqrt{n} \J_{n_{\calMt}}^{-1/2} \P_m \J_{n_{\calMt}}^{1/2}\Q_m\betabs_{\calMt_m}}
    \leq \sqrt{n} \ltwonorm{\J_{n_{\calMt}}^{-1/2} \P_m \J_{n_{\calMt}}^{1/2}\Q_m} \ltwonorm{\betabs_{\calMt_m}} = O_P(n^{\kappa}).
\end{equation}
Notice that $\betabh_{\calMt^{\perp}} = \betabh_{0_{\calMt^{\perp}}} = \zero$ and $\tilde{s}_m \leq \tilde{s} = O_P(n^{2\kappa})$. 
This together with \eqref{eq:diff_bh}, \eqref{eq:diff_bh_1} and \eqref{eq:diff_bh_2} imply that
\begin{equation}
  \label{eq:diff_est} 
  n \| \betabh - \betabh_{0} \|_2^2
  = n \| \betabh_{\calMt} - \betabh_{0_{\calMt}} \|_2^2 
    =O_P(\tilde{s}_m) + O_P(n^{2\kappa}) + o_P(\max\{\tilde{s}_m, n^{2\kappa}\})
    = O_P(n^{2\kappa}).
\end{equation}

By H\"older's inequality,
\begin{equation} \label{eq:diff_est2}
    \|\betabh - \betabh_0\|_1 = \|\betabh_{\calMt} - \betabh_{0_{\calMt}}\|_1 \leq \sqrt{\tilde{s}} \| \betabh_{\calMt} - \betabh_{0_{\calMt}}\|_2 = O_P(n^{2\kappa-\frac{1}{2}})  .
\end{equation}
By the KKT condition, for any $j \in \calMt $, $0 \in -S_n(\betabh)_j + I(j\notin \calMt_m) p'_{\lambda_1}(|\hat{\beta}_j|)$. Since
$|p'_{\lambda_1}(|\hat{\beta}_j|)| \leq \lambda_1$ by the property of folded concave penalty, it follows that
$|S_n(\betabh)_j | = |p'_{\lambda_1}(\hat{\beta}_j)| \leq \lambda_1=O(\sqrt{\log \tilde{s}/n})$. 
Under Condition \ref{cond13}, $\tilde{s} = O(n^{2\kappa}) = o(n^{\frac{1}{3}})$. Then, \eqref{eq:diff_est2} implies that,
\begin{equation}
\label{eq:eq2a}
\begin{split}
  &\hspace{3ex}|(\betabh - \betabh_0)^T \X \{\Y-\boldsymbol{b}'(\X\betabh) \}|
  \leq  \| \betabh_0 - \betabh \|_1 \|  \S_n(\betabh)\|_{\infty}\\
  &= O_P(n^{2\kappa-\frac{1}{2}}) O_P(\sqrt{{(\log \tilde{s})}/{n}}) 
   = O_P(n^{2\kappa-1}\sqrt{\log n}) =  o_P(1).    
\end{split}
\end{equation}

Next, let $\J_{n_{\calMt}}=\J_n(\betabs)_{\calMt}$ and 
\begin{equation}
\label{eq:15}
T_1  =(\betabh_{\calMt} - \betabh_{0_{\calMt}})^T \J_{n_{\calMt}}(\betabh_{\calMt} - \betabh_{0_{\calMt}}).
\end{equation}
We show that $II = T_1 + o_P(1)$. Since for any $j \in \calMt^{\perp}$, $\betah_j = 0 $ and $\betah_{0_j} = 0$, we have
\begin{align*}
    (\betabh - \betabh_0)^T \J_n(\betabh) (\betabh - \betabh_0) = 
    (\betabh_{\calMt} - \betabh_{0_{\calMt}})^T \J_n(\betabh)_{\calMt} (\betabh_{\calMt} - \betabh_{0_{\calMt}}),
\end{align*}
where $\J_n(\hat{\betab})_{\calMt} = {n}^{-1} \sum_{i=1}^{n}b''(\X_i^T\hat{\betab}) \X_{i,\calMt} \X_{i,\calMt}^T $. 
By Taylor's expansion of $\J_n(\betabh)_{\calMt}$ around $\betabs$, we have 
\begin{align*}
    (\betabh_{\calMt} - \betabh_{0_{\calMt}})^T \J_n(\betabh)_{\calMt} (\betabh_{\calMt} - \betabh_{0_{\calMt}})
    &= (\betabh_{\calMt} - \betabh_{0_{\calMt}})^T \J_{n_{\calMt}}(\betabh_{\calMt} - \betabh_{0_{\calMt}}) + 
     \R_a^T {(\betabh_{\calMt}-\betabs_{\calMt})} ,
\end{align*}
where $\R_a$ is a $p$-dim vector whose $j$-th element
$$ R_j = \frac{1}{2n} (\betabh_{\calMt} - \betabh_{0_{\calMt}})^T \bX_{\calMt}^T \diag \{ |X_j| \circ b'''(\bX \Tilde{\betab})\}
\X_{\calMt} (\betabh_{\calMt} - \betabh_{0_{\calMt}}),$$ and $\Tilde{\betab} = t\betabh + (1-t)\betabs$ for some $t\in (0,1)$.
Then,
\begin{equation}
\label{eq:16}
\begin{split}
    &\hspace{3ex}|(\betabh_{\calMt} - \betabh_{0_{\calMt}})^T 
    (\J_n(\betabh)_{\calMt} - \J_{n_{\calMt}}) 
    (\betabh_{\calMt} - \betabh_{0_{\calMt}})| 
    =  |\bR_a^T (\betabh_{\calMt}-\betab_{\calMt}^*)| \\
    & \leq  \frac{1}{2n}\| \betabh_{\calMt} - \betab_{\calMt}^* \|_1  \| \betabh_{\calMt} - \betabh_{0_{\calMt}} \|_2^2 
      \cdot
     \max_{j \in \calMt} \lambda_{\max}(\X_{\calMt} ^T\diag \{ |X_j| \circ b'''(\bX \Tilde{\betab})\} \X_{\calMt}).
\end{split}
\end{equation}
Under Condition~\ref{cond12},
$ \max_{j \in \calMt} \lambda_{\max}(\X_{\calMt}^T\diag \{ |X_j| \circ b'''(\X \Tilde{\betab})\} \X_{\calMt})=
O_P(n)$. Then, it follows from \eqref{Negahban}, \eqref{eq:diff_est} and \eqref{eq:16} that
\begin{equation}
\label{eq:eq2b}
    | (\betabh - \betabh_0)^T \{\J_n(\betabh) - \J_n(\betabs) \} (\betabh - \betabh_0)| 
       = O_P\left( s\sqrt{\frac{\log \tilde{s}}{n}}\right) O_P(n^{2\kappa-1}) O_P(1)
    = o_P(1).
\end{equation}
For $III$, we have
\begin{equation}
  \label{eq:eq2c}
  \begin{split}
    &\hspace{3ex}|(\betabh- \betabh_0)^T \R_n|  \leq 
    \| \betabh_{\calMt}- \betabh_{0_{\calMt}} \|_1 \| \R_n \|_{\infty} \\
    & \leq \frac{1}{n} \| \betabh_{\calMt}-\betabh_{0_{\calMt}} \|_1 \| \betabh_{\calMt} - \betabh_{0_{\calMt}} \|_2^2 \cdot \sup_{j\in \calMt}
      \lambda_{\max} \{ \X_{\calMt} ^T \diag \{ |X_j| \circ b'''(\X \Tilde{\betab})\} \X_{\calMt} \}\\ 
   &= O_P( s\sqrt{{\log \tilde{s}}/{n}}) O_P(n^{2\kappa-1}) O_P(1)
   =  o_P(1).
  \end{split}
\end{equation}
Recalling the definition of $T_{1}$ from (\ref{eq:15}), 
it follows from \eqref{eq:Taylor_Hm}, \eqref{eq:eq2a}, \eqref{eq:eq2b} and \eqref{eq:eq2c} that
\begin{equation}
    \hat{H}_m = \frac{2}{n}\{ \log L_n(\betabh) - \log L_n(\betabh_0) \}= T_1 + o_P(1).
\end{equation}
Next, we derive the asymptotic distribution of $T_n = nT_1$. Note that,  
$
    T_n  = \| \J_{n_{\calMt}}^{1/2} \sqrt{n} (\betabh_{\calMt}  - \betabh_{0_{\calMt}})\|_2^2.
$
Then, from \eqref{eq:diff_bh} we have
\begin{equation*}
  \begin{split}
  T_n &= \| \P_m [ n^{-1/2} \J_{n_{\calMt}}^{-1/2} \X_{\calMt}^T \{\Y-\boldsymbol{b}'(\X \betabs) \} 
        + \sqrt{n} \J_{n_{\calMt}}^{1/2}\Q_m\betabs_{\calMt_m} ]   + o_P(1)\|_2^2 \\
      &= \| n^{-1/2} \K_{mm}^{-1/2} \Q_m^T \J_{n_{\calMt}}^{-1} \X_{\calMt}^T \{\Y-\boldsymbol{b}'(\X \betabs) \} 
        + \sqrt{n} (\K_{mm}^{-1/2} -\Omegab_{mm}^{-1/2} )\betabs_{\calMt_m}  \\
      &\hspace{3ex} + \sqrt{n} \Omegab_{mm}^{-1/2} \betabs_{\calMt_m} + o_P(1)\|_2^2,
    \end{split}
\end{equation*}
where $\K_{mm} = \Q_m^T \J_{n_{\calMt}}^{-1} \Q_m$, $\Omegab_{mm}=\Q_m^T \J_{0_{\calMt}}^{-1} \Q_m$, and $\J_{0_{\calMt}} =
      \E[b''(\x^T\betabs)\x_{\calMt}\x_{\calMt}^T]$.  Let  $\epsilonb_0 =  \sqrt{n} (\K_{mm}^{-1/2} -\Omegab_{mm}^{-1/2}
      )\betabs_{\calMt_m}$. Recall that $\ltwonorm{ \betab_{\calMt_m}^* } = O(n^{\kappa-\frac{1}{2}})$ and
      $\lambda_{\max}(\J_{n_{\calMt}}) = O_P(1)$, we have
\begin{equation*}
  \begin{split}
   \|\epsilonb_0 \|_2^2 &\leq n\|\betabs_{\calMt_m}\|_2^2 \|\K_{mm}^{-1/2} -\Omegab_{mm}^{-1/2}\|_2^2 
    \leq n\|\betabs_{\calMt_m}\|_2^2 \|\K_{mm}^{-1} -\Omegab_{mm}^{-1}\|_2 \\
    &\leq n\|\betabs_{\calMt_m}\|_2^2 \|\K_{mm}^{-1}\|_2 \|\K_{mm} -\Omegab_{mm}\|_2 \|\Omegab_{mm}^{-1}\|_2\\
    &\lesssim n \|\betabs_{\calMt_m}\|_2^2 \|\Q_m^T( \J_{n_{\calMt}}^{-1} - \J_{0_{\calMt}}^{-1} )\Q_m \|_2  \lesssim n \|\betabs_{\calMt_m}\|_2^2 \|\J_{n_{\calMt_m}} - \J_{0_{\calMt_m}} \|_2
    =o_P\left(1 \right), 
    \end{split}
\end{equation*}
since by the matrix Bernstein inequality,  
\begin{equation*}
  \begin{split}
    &\hspace{3ex}\|\J_{n_{\calMt_m}} - \J_{0_{\calMt_m}} \|_2 
    \leq \sqrt{\tilde{s}_m} \|\J_{n_{\calMt_m}} - \J_{0_{\calMt_m}} \|_1     \\
    &= \sqrt{\tilde{s}_m} \|n^{-1}b''(\x_i^T\betabs)\x_{i,\calMt_m}\x_{i,\calMt_m}^T - \E[b''(\x^T\betabs)\x_{\calMt_m}\x_{\calMt_m}^T] \|_1 
    = O_P(\sqrt{n^{-1}{\tilde{s}_m \log \tilde{s}_m}}).
    \end{split}
\end{equation*}
Here, we use $a_n \lesssim b_n$ to denote that there exists a constant $C$ such that $a_n\leq Cb_n$. Hence, 
\begin{equation*}
    T_n = \| n^{-1/2} \K_{mm}^{-1/2} \Q_m^T \J_{n_{\calMt}}^{-1} \X_{\calMt}^T \{\Y-\boldsymbol{b}'(\X \betabs) \} 
    + \sqrt{n} \Omegab_{mm}^{-1/2} \betabs_{\calMt_m} \|_2^2 + o_P(1).
\end{equation*}
Let $\deltab_i = n^{-1/2}\P_m \J_{n_{\calMt}}^{-1/2} \{Y_i - b'(\X_i^T \betabs) \}\X_i $, we have 
\begin{equation*}
  \begin{split}
    \W & = n^{-1/2} \P_m \J_{n_{\calMt}}^{-1/2} \X_{\calMt}^T \{\Y-\boldsymbol{b}'(\X \betabs) \}
    = \sum_{i=1}^n n^{-1/2} \P_m \J_{n_{\calMt}}^{-1/2} \X_i \{Y_i - b'(\X_i^T \betabs) \} 
    = \sum_{i=1}^n \deltab_i. 
    \end{split}
\end{equation*}
Since $\E(y|\x) = b'(\x^T \betabs)$, $\E[\deltab_i] = \zero$. Moreover,
\begin{equation*}
  \begin{split}
    \sum_{i=1}^n \mathrm{Cov}(\deltab_i) &= \sum_{i=1}^n\E(\deltab_i\deltab_i^T)
    = \E\left[ n^{-1} \sum_{i=1}^n  \P_m \J_{n_{\calMt}}^{-1/2} \X_i \{Y_i - b'(\X_i^T
      \betabs) \}^2 \X_i^T \J_{n_{\calMt}}^{-1/2} \P_m \right] \\
    &= \E \left\{  \P_m \J_{n_{\calMt}}^{-1/2} \left[ n^{-1}  \sum_{i=1}^n\X_{i} \E\{Y_i - b'(\X_i^T \betabs) |\X_i \}^2 \X_i^T \right]  \J_{n_{\calMt}}^{-1/2}\P_m  \right\} = \E (\P_m )= \I_{\tilde{s}_m},
    \end{split}
\end{equation*}
\vspace{-1ex} 
\begin{equation*}
  \begin{split}
    \Tilde{s}_m^{1/4} \sum_{i=1}^n \E\| \deltab_i \|_2^3 
       &= n^{-3/2} \Tilde{s}_m^{1/4} \sum_{i=1}^n \E[\| \P_m \J_{n_{\calMt}}^{-1/2} \X_i \|_2^3 \E\{Y_i -
          b'(\X_i^T\betabs)|\X_i\}^3 ] \\ 
        &\lesssim n^{-3/2} \Tilde{s}_m^{1/4} \sum_{i=1}^n \E( \| \P_m \|_2^3 \|
          \J_{n_{\calMt}}^{-1/2} \X_i \|_2^3) 
        \lesssim n^{-3/2} \Tilde{s}_m^{1/4} \sum_{i=1}^n \E(\X_i^T \J_{n_{\calMt}} ^{-1} \X_i)^{3/2} 
        = o(1).
    \end{split}
\end{equation*}
      Then, by the Lemma S6 of \cite{shi2019linear}, we have 
\begin{equation}
    \sup_{\mathcal{C} } |P(\W \in \mathcal{C}) - P(\Z \in \mathcal{C})| \rightarrow 0,
\end{equation}
      where the supremum is taken over all convex set $\mathcal{C} \subset \mathcal{R}^{\tilde{s}_m}$ and $\Z  \sim N_{\tilde{s}_m} (\zero,  \allowbreak \I_{\tilde{s}_m})$.
      Let $T_0 = \| \W + \sqrt{n} \Omegab_{mm}^{-1/2}\betabs_{\calMt_m} \|^2 $ and $C_x = \{ \Z \in
      \mathcal{R}^{\tilde{s}_m} : \| \Z + \sqrt{n} \Omegab_{mm}^{-1/2}\betabs_{\calMt_m} \|_2^2 \leq x \}$. We have $P(\W
      \in C_x ) = P(T_0 \leq x). $ 
      Then,  
\begin{equation*}
    \sup_{x>0} |P(T_0 \leq x) - P(\chi_{\tilde{s}_m,\Gamma_n}^2 \leq x)| \rightarrow 0,
\end{equation*}
      where $\Gamma_n = n \betab_{\calMt_m}^{*T} \Omegab_{mm}^{-1} \betabs_{\calMt_m}$. 
      Since $ \hat{H}_m = T_1
      + o_P(1) = n^{-1}T_0 + o_P(1)$, for any arbitrary $\epsilon>0$, we have  
\begin{equation*}
  \begin{split}
    & \hspace{3ex}P(n^{-1} \chi_{\tilde{s}_m,\Gamma_n}^2 \leq \frac{x}{n} - \epsilon) + o(1) 
     \leq  P(n^{-1} T_0 + \epsilon \leq \frac{x}{n}) + o(1) \\
    & \leq  P(\hat{H}_m \leq \frac{x}{n}) \leq P( n^{-1}T_0-\epsilon \leq \frac{x}{n} )+ o(1) 
     \leq  P(n^{-1} \chi_{\tilde{s}_m,\Gamma_n}^2 \leq \frac{x}{n} + \epsilon), 
    \end{split}
\end{equation*}
      which implies that
$
P(\chi_{\tilde{s}_m,\Gamma_n}^2 \leq x - n\epsilon ) + o(1) \leq P(n \hat{H}_m \leq x)
\leq P( \chi_{\tilde{s}_m,\Gamma_n}^2 \leq x + n\epsilon).
$
Letting $n \epsilon \rightarrow 0$, we have
\begin{equation}
    \sup_{x>0} |P(  n\hat{H}_m\leq x) - P(\chi_{\tilde{s}_m,\Gamma_n}^2 \leq x)| \rightarrow 0.
  \end{equation}

\bigskip
\subsection{Proof of Proposition~\ref{pro3}}
(1) Letting $\Deltab=\betabs - \betab_0$, we have
\begin{equation*}
    H_m = 2 \E[\Deltab^T\x y - b(\x^T\betabs) + b(\x^T\betab_0)].
\end{equation*}
Since $\E[\Deltab^T\x^T\{y-b'(\x^T\betabs) \}]  =  \E[\Deltab^T\x^T \allowbreak 
  \E\{y-b'(\x^T\betabs)|\x \}] = \zero$, it follows from \eqref{eq: Taylor_bfun} that 
\begin{equation}
    \label{eq10_Taylor}
    H_m = 2 \E[\Deltab^T \x^T\{y-b'(\x^T\betabs) \}] + \Deltab^T \J_0(\tilde{\betab}) \Deltab=\Deltab_{\calM}^T \J_0(\tilde{\betab})_{\calM}
    \Deltab_{\calM}, 
\end{equation}
where $\tilde{\betab} = t\betabs+(1-t)\betab_0$ for some $t\in (0,1)$ and $\J_0(\betab) = \E\{ b'' (\x^T \betab)\x\x^T
\}$. Note that $\betabs_{\calM^{\perp}} = \betab_{0_{\calM^{\perp}}} = \zero$, 
$\Deltab_{\calM}=(\Deltab_{\calM_{-m}}^T, \betab_{\calM_{m}}^{*T})^T$, and 
$\|\Deltab_{\calM_{-m}}\|_2^2 = O(\tilde{s}_m) = O(n^{2\kappa})$. Under Condition~\ref{cond11},
$\|\betabs_{\calM_{m}}\|_2^2 = O(n^{2\kappa - 1})$. Under Condition~\ref{cond12},
$\lambda_{\max}(\J_0(\tilde{\betab})_{\calM}) = O(1)$.
Then, we have
\begin{equation} \label{eq:crossterm}
    \begin{split}
   &\hspace{3ex}\betab_{\calM_{m}}^{*T} \J_0(\tilde{\betab})_{\calM_{m}\calM_{-m}} \Deltab_{\calM_{-m}}
    \leq  \E\{(\betab_{\calM_{m}}^{*T} \x_{\calM_{m}}\x_{\calM_{-m}}^T \Deltab_{\calM_{-m}})^2 \}^{1/2}
    \E\{ b''(\x^T\tilde{\betab})^2 \}^{1/2} \\
    &\leq  \E[\betab_{\calM_{m}}^{*T} \x_{\calM_{m}}\x_{\calM_{m}}^T \betabs_{\calM_{m}}
    \Deltab_{\calM_{-m}}^{T} \x_{\calM_{-m}}\x_{\calM_{-m}}^T\Deltab_{\calM_{-m}}]^{1/2} 
    \cdot\E[ b''(\x^T\tilde{\betab})^2]^{1/2} \\
    & =  \E[\|\x_{\calM_{m}}^T \betabs_{\calM_{m}}\|_2^2 \|\x_{\calM_{-m}}^T\Deltab_{\calM_{-m}}\|_2^2]^{1/2} \E[
      b''(\x^T\tilde{\betab})^2]^{1/2} \\ 
    & =  O(n^{2\kappa-1}) O(n^{2\kappa})  O(1)  = o(1). 
    \end{split}
\end{equation}
Since $\Deltab_{\calM_{-m}}^T \J_0(\tilde{\betab})_{\calM_{-m}}\Deltab_{\calM_{-m}} \geq 0$, it follows from
\eqref{eq10_Taylor} and \eqref{eq:crossterm} that
\begin{equation}
    \label{eq:quadratic_ineq}
      \begin{split}
        H_m 
        &= \betab_{\calM_{m}}^{*T} \J_0(\tilde{\betab})_{\calM_{m}}\betabs_{\calM_{m}} + 
        \Deltab_{\calM_{-m}}^{*T} \J_0(\tilde{\betab})_{\calM_{-m}}\Deltab_{\calM_{-m}} + 
        2 \Deltab_{\calM_{-m}}^{*T} \J_0(\tilde{\betab})_{\calM_{-m,m}}\betabs_{\calM_{m}} \\
        &\geq \betab_{\calM_{m}}^{*T} \J_0(\tilde{\betab})_{\calM_{m}}\betabs_{\calM_{m}} + 
        2 \Deltab_{\calM_{-m}}^{*T} \J_0(\tilde{\betab})_{\calM_{-m,m}}\betabs_{\calM_{m}}\\
        &= \betab_{\calM_{m}}^{*T} \J_0(\tilde{\betab})_{\calM_{m}}\betabs_{\calM_{m}} + o(1).
      \end{split}
\end{equation}
Moreover, we have  
\begin{equation}
    \label{eq:submatrix_bound}
    \begin{split}
    \betab_{\calM_{m}}^{*T} \J_0(\tilde{\betab})_{\calM_{m}}\betabs_{\calM_{m}} 
    &\leq  \|\betabs_{\calM_{m}}\|_2^2\cdot \lambda_{\max}(\J_0(\tilde{\betab})_{\calM_{m}}) 
    = O(n^{2\kappa-1}),\\
    \betab_{\calM_{m}}^{*T} (\J_0(\tilde{\betab})_{\calM_{m}}- \J_0(\betabs)_{\calM_{m}} )\betabs_{\calM_{m}}
     &=  \E[|\x_{\calM_{m}}^T\betabs_{\calM_{m}}|^2 \{b''(\x^T\tilde{\betab}) - b''(\x^T\betabs) \}] \\
     &\leq \E[|\x_{\calM_{m}}^T\betabs_{\calM_{m}}|^4]^{1/2} \E[\{b''(\x^T\tilde{\betab}) - b''(\x^T\betabs) \}^2]^{1/2} \\
     &= O(n^{2\kappa-1}) o(1) = o(n^{2\kappa - 1}),
    \end{split}
\end{equation}
where in the second-to-last equation, we utilize Conditions \ref{cond1}, \ref{cond6} and \ref{cond11}, which imply that 
\begin{equation*}
    \begin{split}    
    \E[|\x_{\calM_{m}}^T\betabs_{\calM_{m}}|^4] &=O\left(n^{4\kappa-2}\right),\\
      \E[\{b''(\x^T\tilde{\betab}) - b''(\x^T\betabs) \}^2] & \lesssim \E[\|\x^T (\tilde{\betab} - \betabs)\|^2] =O\left(\frac{s \log n}{n}\right) = o(1).
    \end{split}
\end{equation*}
 Note that $\beta_j^*=0$ for any $j \notin \calM$, thus it follows from \eqref{eq:quadratic_ineq} and
 \eqref{eq:submatrix_bound} that
\begin{equation}
\label{eq:Hm_qual}
  n H_m  \geq n \betab_{\calM_{m}}^{*T} \J_0(\betabs)_{\calM_{m}}\betabs_{\calM_{m}}
        + o(n)
      = n \betab_{\calMt_{m}}^{*T} \J_0(\betabs)_{\calMt_{m}}\betabs_{\calMt_{m}}
        + o(n)
      = \Gamma_n + o(n).
\end{equation}
(2) Under the additional assumption that
$n \|\Deltab_{\calM_{-m}}\|_2^2=n \|\betabs_{-m}-\betab_{0_{-m}}\|_2^2 = o(n^{2\kappa})$, we have,
\begin{equation*}
    \Deltab_{\calM_{-m}}^T \J_0(\tilde{\betab})_{\calM_{-m}}\Deltab_{\calM_{-m}} 
    \leq  \|\Deltab_{\calM_{-m}}\|_2^2 \cdot\lambda_{\max} (\J_0(\tilde{\betab})_{\calM_{-m}}) 
      = o\left(n^{2\kappa-1}\right).
\end{equation*}
Then, 
\begin{equation}
  \label{eq: Taylor_Hm}
  \begin{split}
      H_m &=  \betab_{\calM_{m}}^{*T} \J_0(\tilde{\betab})_{\calM_{m}}\betabs_{\calM_{m}} +
      \Deltab_{\calM_{-m}}^T \J_0(\tilde{\betab})_{\calM_{-m}}\Deltab_{\calM_{-m}}
      + 2 \betab_{\calM_{m}}^{*T} \J_0(\tilde{\betab})_{\calM_{m}\calM_{-m}}
       \Deltab_{\calM_{-m}} \\
    &= \betab_{\calM_{m}}^{*T} \J_0(\betabs)_{\calM_{m}}\betabs_{\calM_{m}} +o\left(n^{2\kappa-1}\right) + o\left(n^{2\kappa-1}\right) \\
    & = \betab_{\calM_{m}}^{*T} \J_0(\betabs)_{\calM_{m}}\betabs_{\calM_{m}} + o(1)
    =  \betab_{\calMt_{m}}^{*T} \J_0(\betabs)_{\calMt_{m}}\betabs_{\calMt_{m}} + o(1).
  \end{split}
\end{equation}

\section{Supporting Lemmas and their proofs}

\begin{lem} \label{lemma3}
Under Conditions~\ref{cond1} -- \ref{cond9},  
$\lambda_{\max}(\J_n(\tilde{\betab})_{\calMt}) = O_P(1)$ for any $\tilde{\betab} \in \mathcal{N}_0$, where $\mathcal{N}_0 = \{ \betab \in \mathcal{R}^p : \| \betab - \betabs\|_2\ \leq C \sqrt{n^{-1} s \log n}, \betab_{\mathcal{\Tilde{M}}^{\perp}}=\zero  \}$ for some $C >0$.
\end{lem}

\begin{proof} 
Note that, 
\begin{equation}
    \lambda_{\max}(\J_n(\tilde{\betab})_{\calMt}) \leq \sup_{\betab \in \mathcal{N}_0} \{
    b''(\X_i^T\betab) \} \cdot \lambda_{\max} \left( \frac{1}{n} \sum_{i=1}^{n} \X_{i,\calMt} \X_{i,\calMt}^T
    \right). 
\end{equation}

First, we show that $\sup_{\betab \in \mathcal{N}_0}b''(\X_i'\betab) = O_P(1)$.
By Condition \ref{cond6} and Chebyshev's inequality, for any $t > 0$ and some $C>0$,
\begin{align*}
  & \hspace{3ex} P\left(\sup_{\betab \in \mathcal{N}_0} |b''(\X_i^T\betab) - b''(\X_i^T\betabs)|
    \geq t\right)  
  \leq P(C |\X_i^T\betab - \X_i^T\betabs| \geq t)\\
  &\lesssim \frac{\E[(\X_i^T\betab - \X_i^T\betabs)^2]}{t^2}
    \lesssim \frac{\lambda_{\max} ( \Sigmab ) \| \betab - \betabs \|_2^2}{t^2} .
\end{align*}
Since $ \| \betab - \betabs \|_2\lesssim \sqrt{n^{-1}s \log n} =o_P(1)$ and $\lambda_{\max}(\Sigmab) =O(1)$, we
have
$P( \sup_{\betab \in \mathcal{N}_0} |b''(\X_i^T\betab) - \allowbreak b''( \X_i^T\betabs)| \geq t)
\rightarrow 0 $.
Then, by the triangle inequality and Condition~\ref{cond9},
\begin{equation} \label{eq:lm3}
  \begin{split}
    \sup_{\betab \in \mathcal{N}_0} |b''(\X_i^T\betab)|
    &\leq |b''(\X_i^T\betabs)| + \sup_{\betab \in \mathcal{N}_0} |b''(\X_i^T\betabs) -
      b''(\X_i^T\betab)|\\
      &= O_P(1) + o_P(1) = O_P(1).
  \end{split}
\end{equation}

Next, we show that $\lambda_{\max}( {n}^{-1} \sum_{i=1}^{n} \X_{i,\calMt} \X_{i,\calMt}^T) = O_P(1)$.  Since $\X_i$ is
sub-Gaussian with parameter $\sigma^2$, so is $\X_{i,\calMt }$.
Let $\Z_i = \X_{i,\calMt} \X_{i,\calMt}^T - \Sigmab_{\calMt}$, which is an $\Tilde{s} \times \Tilde{s}$ symmetric
matrix with $\E[\Z_i] = \zero$. We have
\begin{equation*}
  \begin{split}
 \lambda_{\max}\left( \frac{1}{n} \sum_{i=1}^{n} \X_{i,\calMt} \X_{i,\calMt}^T - \Sigmab_{\calMt}\right)
 &   \leq  \frac{1}{n}\sum_{i=1}^{n} \|  \Z_i \|_2 
     \leq \frac{1}{n} \sum_{i=1}^{n}  \mathbf{u}_i^T \Z_i \mathbf{u}_i
  = \frac{1}{n} \sum_{i=1}^n \{\xi_i^2-\E(\xi_i^2) \},
  \end{split}
\end{equation*}
where $\mathbf{u}_i = \argmax_{\nub \in \Rcal^{\tilde{s}}, \|\nub\|=1} \nub^T \Z_i \nub $ and $ \xi_i =
\mathbf{u}_i^T\X_{i,\calMt}$.  
Then,
\begin{align*}
    \E[\exp \{s \mathbf{u}_i^T \Z_i \mathbf{u}_i \}] &= \E [\exp\{s (\xi_i^2 - \E[\xi_i^2])\}] = 1 + \sum_{k=2}^{\infty}
      \frac{s^k\E[\xi_i^2 - \E [\xi_i^2]]^k}{k!} \\ 
    &\leq 1 + \sum_{k=2}^{\infty} \frac{s^k 2^{k-1}\E[\xi_i^{2k} - \E [\xi_i^2]^k]}{k!} 
    \leq 1 + \sum_{k=2}^{\infty} \frac{s^k 4^k\E[\xi_i^{2k}]}{2k!} 
    \leq 1 + \sum_{k=2}^{\infty} \frac{s^k 4^k (2\sigma^2)^k k!}{k!} \\
    &= 1 + (8s\sigma^2)^2\sum_{k=2}^{\infty}(8s\sigma^2)^k \leq 1 + 128s^2\sigma^4 \leq \exp\{s^2 (16\sigma^2)^2/2\},
\end{align*}
where the first two inequalities follow from Jensen's inequality and the third inequality follows from sub-Gaussianity of
$\xi_i$, and the second to the last inequality holds when $|s|\leq 1/(16\sigma^2)$. That implies that $\mathbf{u}_i^T \Z_i \mathbf{u}_i $ is sub-Exponential with parameter $16 \sigma^2$. 
Then by Berstein's inequality,
$$  P\left(\frac{\tilde{s}}{n} \sum_{i=1}^{n} |\mathbf{u}_i^T \Z_i \mathbf{u}_i | > C\right) \leq 2 \exp \left[-\frac{n}{2} \min\left\{\frac{C}{16\sigma^2\tilde{s}}, \frac{C^2}{(16\sigma^2\tilde{s})^2}\right\}\right],$$
where $C$ is a generic positive constant. 
Then by Condition~\ref{cond1} and the triangle inequality, 
\begin{equation}
      \label{eq:lm3.3}
  \begin{split}
    P\left\{ \lambda_{\max}\left(\frac{1}{n}\sum_{i=1}^n \X_{i,\calMt} \X_{i,\calMt}^T\right) \geq 2 C \right\} 
    & \leq P\left\{ \lambda_{\max}\left(\frac{1}{n}\sum_{i=1}^n \Z_i\right) \geq C \right\}  \\
    & \leq 2 \exp \left[-\frac{n}{2} \min\left\{\frac{C}{16\sigma^2\tilde{s}}, \frac{C^2}{(16\sigma^2\tilde{s})^2}\right\} \right] ,
  \end{split}
\end{equation}
which implies that $\lambda_{\max}(n^{-1}\sum_{i=1}^n \X_{i,\calMt} \X_{i,\calMt}^T)=O_P\left(1 \right) $. 
 Combining (\ref{eq:lm3}) and (\ref{eq:lm3.3}), we have
\begin{align*}
   \lambda_{\max} \left( \J_n(\Tilde{\betab})_{\calMt} \right) 
   \leq \sup_{\betab \in \mathcal{N}_0} b''(\X_i^T\betab) \cdot \lambda_{\max}\left( n^{-1} \sum_{i=1}^{n} \X_{i,\calMt} \X_{i,\calMt}^T\right)  = O_P(1).
\end{align*}

\end{proof}

\begin{lem} \label{lemma5}
Under Conditions \ref{cond1}, \ref{cond9}, \ref{cond10}, and \ref{cond12}, we have
\begin{align*}
    \sqrt{n}(\betabh_{\calMt} - \betab_{\calMt}^*) 
    &= n^{-1/2} \J_{n_{\calMt}}^{-1}\X_{\calMt}^T\{\Y-\boldsymbol{b}'(\X\betabs) \} + o_P(1);\\
    \sqrt{n}(\betabh_{0_{\calMt}} - \betab_{\calMt}^*) 
    &= n^{-1/2}\J_{n_{\calMt}}^{-{1}/{2}} (\I- \P_m)\J_{n_{\calMt}}^{-{1}/{2}} \X_{\calMt}^T\{\Y-\boldsymbol{b}'(\X\betabs) \} 
     - \sqrt{n}\J_{n_{\calMt}}^{-{1}/{2}}\P_m\J_{n_{\calMt}}^{{1}/{2}} \Q_m \betabs_{\calMt_m}+ o_P(1),
\end{align*}
where
$\J_{n_{\calMt}} = \J_n(\betabs)_{\calMt} = n^{-1} \X_{\calMt}^T \diag\{\boldsymbol{b}''(\X \betabs)\} \X_{\calMt}\in
\mathcal{R}^{\tilde{s}\times \tilde{s}}$, $\Q_m=(\zero, \I_{\calMt_m})^T \in \mathcal{R}^{\tilde{s} \times \tilde{s}_m}$,
$\I_{\calMt_m}$ is the $\tilde{s}_m$-dim identity matrix, and
$\P_m =\J_{n_{\calMt}}^{-1/2} \Q_m  \allowbreak (\Q_m \J_{n_{\calMt}}^{-1} \Q_m^T)^{-1} \Q_m^T \J_{n_{\calMt}}^{-1/2}\in
\mathcal{R}^{\tilde{s}\times \tilde{s}}$.

\end{lem}

\begin{proof}
  First, for the asymptotic expansion of $\sqrt{n}(\betabh_{\calMt} - \betab_{\calMt}^*) $, by the Taylor expansion of
  $\S_n(\betabh)_{\calMt}$ around $\betab_{\calMt}^*$, we have
\begin{equation*}
    \S_n(\betabh)_{\calMt} = \S_n(\betabs)_{\calMt} - \J_{n_{\calMt}}(\betabh_{\calMt} - \betab_{\calMt}^*) + \R_n/2,
\end{equation*}
where $\R_n = (R_{n,1}, \dots, R_{n,\tilde{s}})^T$ with
$R_{n,j} = n^{-1}(\betabh_{\calMt} - \betab_{\calMt}^*)^T\X_{\calMt}^T \diag \{ |X_j| \circ b'''(\X \Bar{\betab})\}
\X_{\calMt} \allowbreak (\betabh_{\calMt} - \betab_{\calMt}^*)$ for some $\Bar{\betab} = t\betabh + (1-t)\betabs$ with
$t\in (0,1)$. Then,
\begin{equation}
\label{eq:6}
\begin{split}
  \sqrt{n}(\betabh_{\calMt} - \betab_{\calMt}^*) 
  &= \frac{1}{\sqrt{n}} \J_{n_{\calMt}}^{-1} \{n \S_n(\betabs)_{\calMt}  +     \J_{n_{\calMt}}^{-1}\{-n \S_n(\betabh)_{\calMt}+ \frac{n}{2}\R_n  \}    \} \\
  & = \frac{1}{\sqrt{n}} \J_{n_{\calMt}}^{-1} \X_{\calMt}^T \{\Y-\boldsymbol{b}'(\X\betabs) \} + \frac{1}{\sqrt{n}}
    \J_{n_{\calMt}}^{-1}\{-n \S_n(\betabh)_{\calMt}+ \frac{n}{2}\R_n  \}.   
\end{split}
\end{equation}
For $j \in \calMt_{-m}$, by the KKT condition of (\ref{eq:9}) and Condition \ref{cond10},
$0 \in S_n(\betabh)_j + p'_{\lambda_1}(\beta_j)$ so that
$|S_n(\betabh)_j| \leq |\lambda_1| = O_P(\sqrt{(\log \tilde{s})/{n}})$. 
For $ j \in \calMt_{m}$, by the KKT condition of (\ref{eq:9}), $S_n(\betabh)_j = 0$. For $\R_n$, under
Condition~\ref{cond12}, we have
\begin{equation*}
    \begin{aligned}
      n R_{n,j}
      &= (\betabh_{\calMt} - \betab_{\calMt}^*)^T \X_{\calMt}^T \diag \{ |X_j| \circ b'''(\X \Bar{\betab})\}
        \X_{\calMt}(\betabh_{\calMt} - \betab_{\calMt}^*) \\ 
      & \leq \| \betabh_{\calMt} - \betab_{\calMt}^* \|_2^2 \cdot \lambda_{\max} ( \X_{\calMt}^T\diag \{ |X_j| \circ
        b'''(\X \Bar{\betab})\}\X_{\calMt}) \\ 
      &= O_P(\frac{ s \log \tilde{s} }{ n})\cdot O_P(n) 
      = O_P(s \log \tilde{s}), 
    \end{aligned}
\end{equation*}
where $\Bar{\betab} = t \betabh + (1-t) \betabs$ for some $t\in (0,1)$ so that $\Bar{\betab} \in \mathcal{N}_0$. Hence,
$\| n \S_n(\betabh)_{\calMt} + n\bR_n \|_{\infty} = O_P(s \log \tilde{s})$. Since
$\lambda_{\max} (\J_{n_{\calMt}}) = O_P(1)$, $\tilde{s} = O_P(n^{2 \kappa})$, and $s = o(n^{1/2})$,
\begin{equation}
    \label{eq:Rn}
    \frac{1}{\sqrt{n}} \J_{n_{\calMt}}^{-1} [ -n \S_n(\betabh)_{\calMt} + \frac{n}{2}\R_n ] = \frac{1}{\sqrt{n}} O_P(1) \cdot O_P(s
    \log \tilde{s}) = o_P(1). 
  \end{equation}
This together with (\ref{eq:6}) imply that 
$
    \sqrt{n}(\betabh_{\calMt} - \betab_{\calMt}^*) = \frac{1}{\sqrt{n}} \J_{n_{\calMt}}^{-1} \X_{\calMt}^T
    \{\Y-\boldsymbol{b}'(\X\betabs) \} + o_P(1) .
$

For $ \sqrt{n}( \betabh_{0_{\calMt}} - \betab_{\calMt}^*) $, by the Taylor expansion of
$\S_n(\betabh_0)_{\calMt}$, we have
\begin{equation*}
  \S_n(\betabh_0)_{\calMt} = \S_n(\betabs)_{\calMt} - \J_{n_{\calMt}}(\betabh_{0_{\calMt}} - \betab_{\calMt}^*) + \frac{1}{2} \R_{n2},
\end{equation*}
where $\R_{n2} = (R_{n2,1}, \dots, R_{n2,\tilde{s}})^T$ with
$R_{n2,j} = n^{-1}(\betabh_{0_{\calMt}} - \betab_{\calMt}^*)^T \X_{\calMt}^T \diag \{\allowbreak |X_j| \circ b'''(\X \Tilde{\betab})\} \allowbreak
\X_{\calMt} (\betabh_{0_{\calMt}} - \betab_{\calMt}^*)$, and $\Tilde{\betab} = t\betabh + (1-t)\betabs$ for some
$t\in (0,1)$. Then,
\begin{equation}
    \label{lm1}
    \sqrt{n}(\betabh_{0_{\calMt}} - \betab_{\calMt}^*) = \frac{1}{\sqrt{n}} \J_{n_{\calMt}}^{-1} \{n \S_n(\betabs)_{\calMt} -  n \S_n(\betabh_0)_{\calMt} +  \frac{n}{2}\R_{n2} \}.
\end{equation}
For $j \in \calMt_{-m}$ , by the KKT condition of (\ref{eq:11}), $0 \in S_n(\betabh_0)_j + p'_{\lambda_2}(\beta_j)$ so that
$|S_n(\betabh_0)_j| \leq \lambda_2 = O_P(\sqrt{\log \tilde{s}/n})$. 
For $j \in \calMt_m$, 
$\betah_{0_j} = 0$. Hence,
$\Q_m^T (\betabh_{0_{\calMt}} - \betab_{\calMt}^*)  = -\betabs_{\calMt_m} $.
Then, we have
\begin{align*}
    -\sqrt{n} \betabs_{\calMt_m} &=  \frac{1} {\sqrt{n}} 
   \Q_m^T \J_{n_{\calMt}}^{-1} \{n \S_n(\betabs)_{\calMt} -  n \S_n(\betabh_0)_{\calMt} + \frac{n}{2}\R_{n2} \}\\
     &=  \frac{1} {\sqrt{n}} 
   \Q_m^T \J_{n_{\calMt}}^{-1} \X_{\calMt}^T (\Y-\boldsymbol{b}'(\X\betabs)) 
      -  {\sqrt{n}} 
   \Q_m^T \J_{n_{\calMt}}^{-1} \Q_m \S_n(\betabh_0)_{\calMt_m} + \bR_{n3},
\end{align*}
where $\R_{n3} = - n^{-1/2} \Q_m^T \J_{n_{\calMt}}^{-1} \{\Q_{-m} n\S_n(\betabh_0)_{\calMt_{-m}} + n\R_{n2}/2 \}$,
$\Q_{-m}=(\zero, \I_{\calMt_{-m}})^T \in \mathcal{R}^{\tilde{s}\times (\tilde{s}-\tilde{s}_m)}$, and
$\I_{\calMt_{-m}}$ is the $(\tilde{s}-\tilde{s}_m)$-dim identity matrix. 
  Similarly as in \eqref{eq:Rn}, we can show that $\R_{n3} = o_p(1)$. Let $\K_{mm} = \Q_m^T \J_{n_{\calMt}}^{-1} \Q_m$. We have
\begin{align*}
    -\sqrt{n} \betabs_{\calMt_m} &= \frac{1} {\sqrt{n}} 
   \Q_m^T \J_{n_{\calMt}}^{-1} \X_{\calMt}^T \{\Y-\boldsymbol{b}'(\X\betabs) \}
     -{\sqrt{n}} \K_{mm}  \Q_m^T \S_n(\betabh_0)_{\calMt}  + o_P(1); \\
    {\sqrt{n}} \Q_m^T \S_n(\betabh_0)_{\calMt} &= {\sqrt{n}}\K_{mm}^{-1} 
    \Q_m^T \J_{n_{\calMt}}^{-1} \S_n(\betabs)_{\calMt} + \sqrt{n} \K_{mm}^{-1} \betabs_{\calMt_m} + o_P(1).
\end{align*}
Plugging them into \eqref{lm1}, we have
\begin{align*}
    & \hspace{3ex} \sqrt{n}(\betabh_{0_{\calMt}} - \betab_{\calMt}^*) \\
    &= {\sqrt{n}} \J_{n_{\calMt}}^{-1} \S_n(\betabs)_{\calMt} 
    - {\sqrt{n}} \J_{n_{\calMt}} ^{-1}\Q_m\Q_m^T \S_n(\betabh_0)_{\calMt} 
    - {\sqrt{n}} \J_{n_{\calMt}} ^{-1}\Q_{-m}\Q_{-m}^T \S_n(\betabh_0)_{\calMt} + \frac{n}{2}\R_{n2}\\
    & = \frac{1}{\sqrt{n}}\J_{n_{\calMt}}^{-{1}/{2}} (\I- \P_m)\J_{n_{\calMt}}^{-{1}/{2}} \X_{\calMt}^T \{\Y-\boldsymbol{b}'(\X\betabs) \} 
    - \sqrt{n}\J_{n_{\calMt}}^{-{1}/{2}}\P_m\J_{n_{\calMt}}^{{1}/{2}} \Q_m \betabs_{\calMt_m}+ o_P(1),
\end{align*}
where $\P_m = \J_{n_{\calMt}}^{-1/2} \Q_m (\Q_m^T \J_{n_{\calMt}}^{-1} \Q_m )^{-1} \Q_m^T \J_{n_{\calMt}}^{-1/2}$
is an $\tilde{s} \times \tilde{s}$ projection matrix.
\end{proof}







\end{document}